\documentclass[sigconf]{acmart}

\AtBeginDocument{%
  }

\copyrightyear{2024}
\acmYear{2024}
\setcopyright{acmlicensed}\acmConference[CCS '24]{Proceedings of the 2024 ACM SIGSAC Conference on Computer and Communications Security}{October 14--18, 2024}{Salt Lake City, UT, USA}
\acmBooktitle{Proceedings of the 2024 ACM SIGSAC Conference on Computer and Communications Security (CCS '24), October 14--18, 2024, Salt Lake City, UT, USA}
\acmDOI{10.1145/3658644.3670398}
\acmISBN{979-8-4007-0636-3/24/10}

\usepackage{amsbsy, amsfonts, amsthm, units, bm}
\usepackage{enumitem}
\usepackage{booktabs}
\usepackage{diagbox}
\usepackage{multirow}
\usepackage{algorithm}
\usepackage{algpseudocode}

\newcommand{\partitle}[1]{ \noindent \textbf{#1.}}
\newcommand{\eqdef}{\mathrel{\coloneqq}}
\newcommand{\projectname}{{\tt ERASER}}

\begin{document}

\newcolumntype{C}[1]{>{\centering\let\newline\\\arraybackslash\hspace{0pt}}m{#1}}
\newcolumntype{K}[1]{>{\centering\arraybackslash}p{#1}}


\makeatletter
\newsavebox\myboxA
\newsavebox\myboxB
\newlength\mylenA

\newcommand*\xbar[2][0.75]{%
    \sbox{\myboxA}{$\m@th#2$}%
    \setbox\myboxB\null
    \ht\myboxB=\ht\myboxA%
    \dp\myboxB=\dp\myboxA%
    \wd\myboxB=#1\wd\myboxA
    \sbox\myboxB{$\m@th\overline{\copy\myboxB}$}
    \setlength\mylenA{\the\wd\myboxA}
    \addtolength\mylenA{-\the\wd\myboxB}%
    \ifdim\wd\myboxB<\wd\myboxA%
       \rlap{\hskip 0.5\mylenA\usebox\myboxB}{\usebox\myboxA}%
    \else
        \hskip -0.5\mylenA\rlap{\usebox\myboxA}{\hskip 0.5\mylenA\usebox\myboxB}%
    \fi}
\makeatother

\renewcommand{\algorithmicrequire}{\textbf{Input:}}
\renewcommand{\algorithmicensure}{\textbf{Output:}}

\newcommand{\tens}[1]{\boldsymbol{\mathscr{#1}}}
\newcommand{\vect}[1]{\ensuremath{\mathbf{#1}}}
\newcommand{\mat}[1]{\ensuremath{\mathbf{#1}}}
\newcommand{\dd}{\mathrm{d}}
\newcommand{\grad}{\nabla}
\newcommand{\hess}{\nabla^2}
\newcommand{\pgrad}[1]{ \frac{\partial}{\partial{#1} }}
\newcommand{\argmin}{\mathop{\rm argmin}}
\newcommand{\argmax}{\mathop{\rm argmax}}
\newcommand{\argsup}{\mathop{\rm argsup}}

\newcommand{\fft}{ \mbox{\tt fft} }
\newcommand{\ifft}{ \mbox{\tt ifft} }
\newcommand{\svd}{ \mbox{\tt svd} }

\newcommand{\Unif}{\textrm{Unif}}
\newcommand{\abs}[1]{\left|{#1}\right|}
\newcommand{\norm}[1]{\left\|{#1}\right\|}
\newcommand{\fnorm}[1]{\|{#1}\|_{\text{F}}}
\newcommand{\infnorm}[1]{\|{#1}\|_\infty}
\newcommand{\spnorm}[2]{\left\| {#1} \right\|_{\text{S}({#2})}}
\newcommand{\sigmin}{\sigma_{\min}}
\newcommand{\tr}{\text{tr}}
\newcommand{\defeq}{\stackrel{\textrm{def}}{=}}
\renewcommand{\det}{\text{det}}
\newcommand{\rank}{\text{rank}}
\newcommand{\logdet}{\text{logdet}}
\newcommand{\trans}{^{\top}}
\newcommand{\poly}{\text{poly}}
\newcommand{\polylog}{\text{polylog}}
\newcommand{\st}{\text{s.t.~}}
\newcommand{\proj}{\mathcal{P}}

\newcommand{\Z}{\mathbb{Z}}
\newcommand{\N}{\mathbb{N}}
\newcommand{\R}{\mathbb{R}}
\newcommand{\Rd}{\mathbb{R}^d}
\newcommand{\Rdd}{\mathbb{R}^{d\times d}}
\newcommand{\Var}{\text{Var}}

\newcommand{\expect}{\mathbb{E}}
\newcommand{\prob}{\mathbb{P}}

\newcommand{\tA}{\tens{A}}
\newcommand{\tT}{\tens{T}}
\newcommand{\tB}{\tens{B}}
\newcommand{\tC}{\tens{C}}
\newcommand{\tD}{\tens{D}}
\newcommand{\tX}{\tens{X}}
\newcommand{\tL}{\tens{L}}
\newcommand{\tE}{\tens{E}}
\newcommand{\tU}{\tens{U}}
\newcommand{\tV}{\tens{V}}
\newcommand{\tS}{\tens{S}}
\newcommand{\tVH}{\tens{V}^{\mathbf{H}}}
\newcommand{\tM}{\tens{M}}
\newcommand{\tJ}{\tens{J}}
\newcommand{\tG}{\tens{G}}
\newcommand{\tI}{\tens{I}}
\newcommand{\tQ}{\tens{Q}}

\newcommand{\tO}{\tens{O}}

\newcommand{\tY}{\tens{Y}}

\newcommand{\td}{{\tt D}_{\w\w}}

\newcommand{\fracpar}[2]{\frac{\partial #1}{\partial  #2}}

\newcommand{\la}{\langle}
\newcommand{\ra}{\rangle}

\newcommand{\X}{\mat{X}}
\renewcommand{\S}{\mat{S}}

\renewcommand{\O}{\mat{O}}
\newcommand{\Q}{\mat{Q}}
\newcommand{\E}{\mat{E}}
\renewcommand{\H}{\mat{H}}






\newcommand{\vectsigma}{\bm{\sigma}}

\newcommand{\mLambda}{\mat{\Lambda}}
\newcommand{\e}{\vect{e}}
\renewcommand{\u}{\vect{u}}
\renewcommand{\v}{\bm{\nu}}
\newcommand{\vt}{\bm{\nu}_{t}}
\newcommand{\vtpp}{\bm{\nu}_{t+1}}
\newcommand{\p}{\vect{p}}
\newcommand{\g}{\vect{g}}
\renewcommand{\a}{\vect{a}}
\newcommand{\w}{\bm{\omega}}
\newcommand{\wu}{\bm{\omega}^u}
\newcommand{\vdelta}{\bm{\delta}}
\newcommand{\wt}{\bm{\omega}_{t}}
\newcommand{\wtpp}{\bm{\omega}_{t+1}}
\newcommand{\wk}{\bm{\omega}^{k-1}}
\newcommand{\wkpp}{\bm{\omega}^{k}}
\newcommand{\wtk}{\bm{\omega}_t^{k-1}}
\newcommand{\x}{\vect{x}}
\newcommand{\y}{\vect{y}}
\newcommand{\z}{\vect{z}}
\newcommand{\fE}{\mathfrak{E}}
\newcommand{\fF}{\mathfrak{F}}

\newcommand{\ur}{\tt{u}}
\newcommand{\UR}{\tt{U}}

\newcommand{\qr}{\tt{q}}
\newcommand{\QR}{\tt{Q}}

\newcommand{\FS}{F_{\tens{S}}}
\newcommand{\FE}{F_{\mathbb{E}}}

\newcommand{\RA}{\bm{\mathcal{A}}}

\newcommand{\cn}{\kappa}
\newcommand{\nn}{\nonumber}
\newcommand{\order}[1]{O(#1)}
\title{{\tt ERASER}: Machine Unlearning in MLaaS via an~Inference~Serving-Aware~Approach}



\author{Yuke Hu}
\orcid{0000-0001-5780-6898}
\authornote{Co-first authors}
\affiliation{%
  \institution{Zhejiang University}
  \department{The State Key Laboratory of Blockchain and Data Security}
  \city{Hangzhou}
  \country{China}
}
\email{yukehu@zju.edu.cn}

\author{Jian Lou}
\orcid{0000-0002-4110-2068}
\authornotemark[1]
\authornotemark[2]
\affiliation{%
  \institution{Zhejiang University}
  \department{The State Key Laboratory of Blockchain and Data Security}
  \city{Hangzhou}
  \country{China}
}
\email{jian.lou@zju.edu.cn}

\author{Jiaqi Liu}
\orcid{0009-0000-7927-1871}
\affiliation{%
  \institution{Zhejiang University}
  \department{The State Key Laboratory of Blockchain and Data Security}
  \city{Hangzhou}
  \country{China}
}
\email{jiaqi.liu@zju.edu.cn}

\author{Wangze Ni}
\orcid{0000-0003-1438-1345}
\affiliation{%
  \institution{Zhejiang University}
  \department{The State Key Laboratory of Blockchain and Data Security}
  \city{Hangzhou}
  \country{China}
}
\affiliation{%
  \institution{Hong Kong University of Science and Technology}
  \city{Hong Kong}
  \country{China}
}
\email{wniab@connect.ust.hk}

\author{Feng Lin}
\orcid{0000-0001-5240-5200}
\affiliation{%
  \institution{Zhejiang University}
  \department{The State Key Laboratory of Blockchain and Data Security}
  \city{Hangzhou}
  \country{China}
}
\email{flin@zju.edu.cn}

\author{Zhan Qin}
\orcid{0000-0001-7872-6969}
\authornote{Corresponding authors}
\affiliation{%
  \institution{Zhejiang University}
  \department{The State Key Laboratory of Blockchain and Data Security}
  \city{Hangzhou}
  \country{China}
}
\email{qinzhan@zju.edu.cn}

\author{Kui Ren}
\orcid{0000-0003-3441-6277}
\affiliation{%
  \institution{Zhejiang University}
  \department{The State Key Laboratory of Blockchain and Data Security}
  \city{Hangzhou}
  \country{China}
}
\email{kuiren@zju.edu.cn}

\renewcommand{\shortauthors}{Hu et al.}
\begin{abstract}
    Over the past years, Machine Learning-as-a-Service (MLaaS) has received a surging demand for supporting Machine Learning-driven services to offer revolutionized user experience across diverse application areas. MLaaS provides inference service with low inference latency based on an ML model trained using a dataset collected from numerous individual data owners. Recently, for the sake of data owners' privacy and to comply with the ``right to be forgotten (RTBF)'' as enacted by data protection legislation, many machine unlearning methods have been proposed to remove data owners' data from trained models upon their unlearning requests. However, despite their promising efficiency, almost all existing machine unlearning methods handle unlearning requests independently from inference requests, which unfortunately introduces a new security issue of inference service obsolescence and a privacy vulnerability of undesirable exposure for machine unlearning in MLaaS. 

    In this paper, we propose the \projectname ~framework for machin\underline{E} unlea\underline{R}ning in MLa\underline{AS} via an inferenc\underline{E} se\underline{R}ving-aware approach. 
    \projectname ~strategically chooses appropriate unlearning execution timing to address the inference service obsolescence issue. A novel inference consistency certification mechanism is proposed to avoid the violation of RTBF principle caused by postponed unlearning executions, thereby mitigating the undesirable exposure vulnerability.
    \projectname ~offers three groups of design choices to allow for tailor-made variants that best suit the specific environments and preferences of various MLaaS systems. 
    Extensive empirical evaluations across various settings confirm \projectname's effectiveness, e.g., it can effectively save up to 99\% of inference latency and 31\% of computation overhead over the inference-oblivion baseline.

\end{abstract}




\maketitle

\section{Introduction}
\label{sec.introduction}
\begin{figure}[t]
    \centering
    \includegraphics[width=0.43\textwidth]{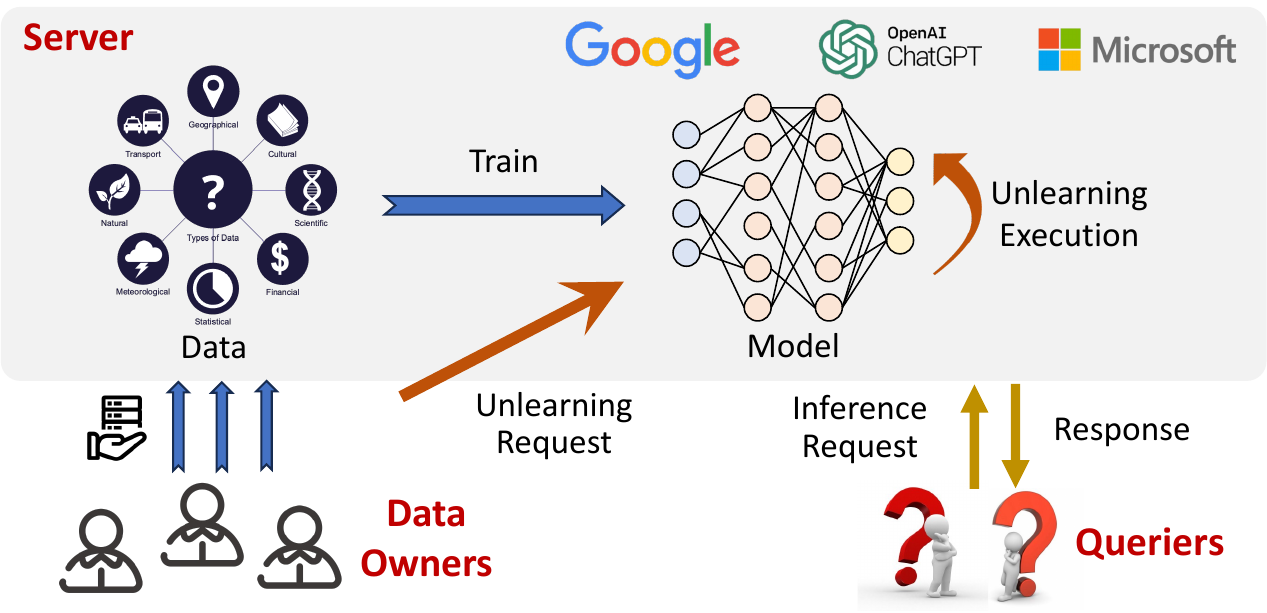}
    \caption{Scheme of MLaaS}
    \Description{The figure shows a server collecting various types of data from data owners.  This data is used to train machine learning models, represented by a neural network diagram. Companies like Google, OpenAI, and Microsoft utilize these models to process inference requests from queriers and provide responses. There is also a process for handling unlearning requests to update the models as needed for privacy and other concerns.}
    \label{fig: mlaas}
 \vspace{-1.5em}
\end{figure}

\begin{figure*}[t]
	\centering
	\includegraphics[width=0.96\textwidth]{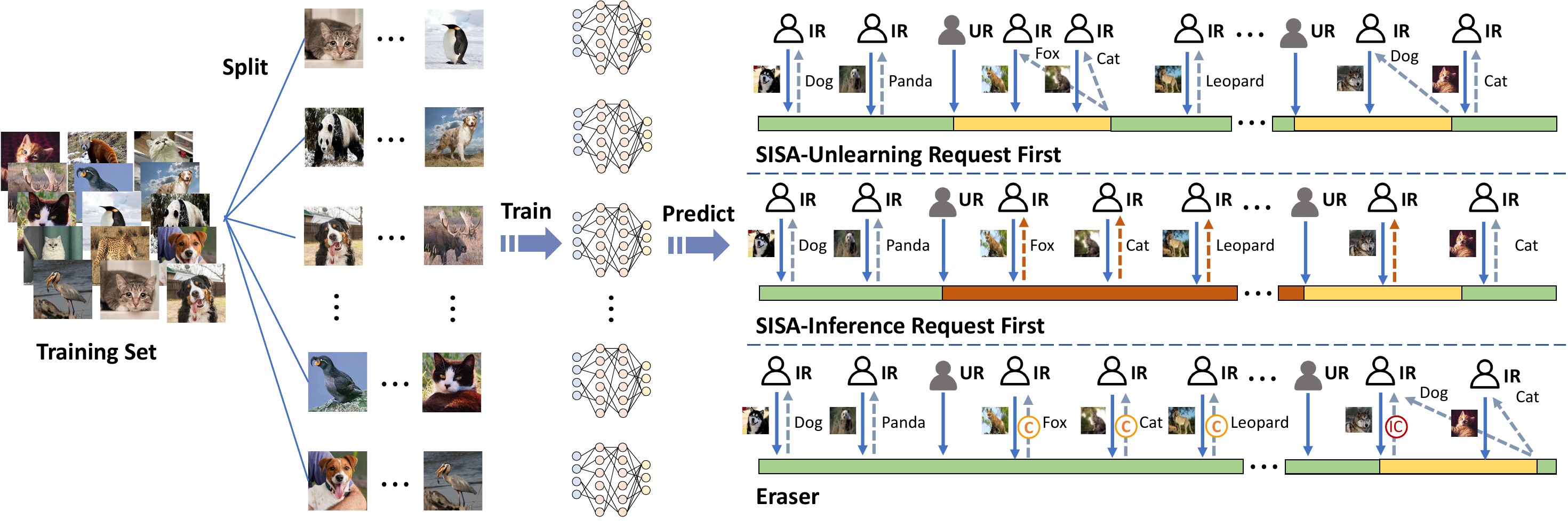}
	\caption{Comparison between \projectname ~and naive strategies. Green represents inference without privacy risk, red indicates inference with privacy risk, and yellow signifies unlearning execution. \textbf{IR}: \textbf{I}nference \textbf{R}equest; \textbf{UR}: \textbf{U}nlearning \textbf{R}equest.}
        \Description{This figure compares the ERASER strategy with naive strategies in the context of machine learning models and their handling of inference and unlearning requests. On the left side, the training set consists of various images. This training set is split into multiple subsets, each used to train separate neural network models. The process involves training the models with these subsets and then using them for predictions. The right side of the figure illustrates the handling of inference and unlearning requests. The comparison is divided into three parts: 1.	SISA-Unlearning Request First: This section shows how inference and unlearning requests are handled by the SISA strategy, prioritizing unlearning requests (UR) before inference requests (IR). Green bars indicate inference without privacy risk, red bars indicate inference with privacy risk, and yellow bars signify unlearning execution. 2.	SISA-Inference Request First: This section depicts the handling of requests when inference requests are prioritized before unlearning requests. Similar to the previous section, green, red, and yellow bars indicate inference without privacy risk, inference with privacy risk, and unlearning execution, respectively.3.	Eraser: This final section represents a naive strategy where inference requests are processed without considering privacy risks or unlearning executions.The comparison highlights the effectiveness of the SISA strategy in managing privacy risks associated with inference requests, as indicated by the varying lengths and colors of the bars in each section.}
	\label{fig: scheme}
	\vspace{-1em}
\end{figure*}

Machine Learning (ML) has witnessed a meteoric rise in its applications across diverse areas \cite{hazelwood2018applied,shaya2010intelligent,bartlett2005recognizing,hauswald2015sirius}, and the ML-powered applications are oftentimes supported by ML-as-a-Service (MLaaS) \cite{SageMaker2018,Tensorflow2018,Nvidia2020} on the \emph{server} side (e.g., cloud platforms \cite{aws2018,azure2018,cloud2018}), whereby the ML model is built on a large training dataset collected from numerous individuals (i.e., \emph{data owners}) and then deployed to serve inference requests from application users (i.e., \emph{queriers}) as depicted in Figure \ref{fig: mlaas}. 
The server's primary objective is to provide queriers with timely responses to their inference requests, keeping inference latency within the bounds of service-level objectives (SLOs) to ensure a satisfactory user experience. \cite{DBLP:conf/nsdi/CrankshawWZFGS17,DBLP:conf/usenix/ZhangYWY19,DBLP:conf/osdi/GujaratiKAHKVM20,gupta2020deeprecsys,DBLP:conf/usenix/Romero0YK21,gunasekaran2022cocktail}. Failure to maintain an acceptable latency could result in clients turning to competitors or incurring penalties. 

Equally crucial, yet previously less explored, is that data owners can submit requirements to revoke their consent for the use of their data due to privacy concerns, demanding the server to delete their data and erase its influence on the trained model. This could be driven by fears that their data is susceptible to privacy attacks via model inferences, such as black-box membership inference attacks (MIA) \cite{li2021membership,liu2022ml, DBLP:conf/iclr/WatsonGCS22, Zhang_2024_WACV, chen2022amplifying, long2018understanding, long2020pragmatic, liu2022membership, liu2022ml, li2021membership}, which deduce the presence of training samples based on inference results. This right is upheld by data protection regulations worldwide, including the European Union's General Data Protection Regulation (GDPR) \cite{mantelero2013eu}, Luxembourg's Commission Nationale pour la Protection des Données (CNPD)\cite{CNPD}, and the California Consumer Privacy Act (CCPA) \cite{California2023}.

Machine unlearning is a nascent research field dedicated to the removal of data and its associated influence from trained ML models in response to unlearning requests, which has gained increasing research attention in recent years \cite{cao2015towards, guo2019certified, sekhari2021remember, graves2021amnesiac, bourtoule2021machine,DBLP:conf/ndss/WarneckePWR23, chen2021machine}. Most existing researches on machine unlearning focus on improving the efficiency of the machine unlearning mechanism and have achieved promising acceleration over the baseline approach of retraining the ML model on the remaining dataset from scratch \cite{cao2015towards, ginart2019making, brophy2021machine, bourtoule2021machine, ullah2021machine, schelter2021hedgecut, chen2022graph, chen2022recommendation,golatkar2020eternal, wu2020deltagrad, golatkar2020forgetting, nguyen2020variational, neel2021descent, peste2021ssse, golatkar2021mixed, warnecke2021machine, izzo2021approximate, mahadevan2021certifiable}. 
Among the many existing approaches, SISA (Sharded, Isolated, Sliced, and Aggregated) \cite{bourtoule2021machine} represents an exact unlearning method that achieves a forgetting effect fully equivalent to retraining from scratch. 

However, existing machine unlearning approaches mostly treat unlearning requests independent from inference requests. This inference-oblivious characteristic precludes straightforward implementation in the MLaaS scenarios. As will be detailed in Sec.\ref{sec.definition.and.threat.model}, we observe that if servers adopt naive strategies for scheduling unlearning and inference requests, it can lead to the security issue of inference service obsolescence and the privacy issue of undesirable exposure.
For illustration, one such naive approach is the ``unlearning-request-first'' strategy, depicted in the top row of Figure \ref{fig: scheme}. Upon receiving unlearning requests, the server temporarily suspends processing inference requests to prioritize unlearning. While this method rigorously complies with the RTBF regulatory mandates, it could result in intolerable response latency, leading to service obsolescence. Conversely, the ``inference-request-first'' strategy, illustrated in the middle row of Figure \ref{fig: scheme}, involves suspending all unlearning requests to prioritize inference requests. Unlearning is deferred until certain conditions are met. 
Another variant of the ``inference-request-first'' strategy is to process the unlearning requests in the backend rather than suspending them while offering inference services by the original model.
Both variants of the second strategy lead to the privacy issue of undesirable exposure. They persist in using the yet-to-be-unlearned model to provide inference results, which still possesses the data that should have been deleted after receiving unlearning requests. This practice inherently breaches the RTBF principle, which mandates ``\textbf{immediately erasure}'' \cite{GDPRRightToBeForgotten, CNPDDroitOubli}\footnote{GDPR mandates that ``personal data must be erased \textbf{immediately} where $\cdots$ withdrawn his consent" \cite{GDPRRightToBeForgotten} and CNPD mandates that ``this right to erasure also allows you to request the \textbf{immediate} removal of personal data" \cite{CNPDDroitOubli}. }.  Additionally, this ``inference-request-first'' strategy keeps to-be-deleted data under potential privacy threats like MIA, which can be otherwise avoided if unlearning is executed timely after being requested.

To address the issues mentioned above, we propose \projectname: machin\underline{\textbf{E}} unlea\underline{\textbf{R}}ning in MLa\underline{\textbf{AS}} via an inferenc\underline{\textbf{E}} se\underline{\textbf{R}}ving-aware approach. \projectname ~can mitigate the issue of inference service obsolescence by strategically postponing unlearning execution through a certified inference consistency mechanism.
The underlying cause of the privacy risk associated with postponing unlearning execution stems from the fact that the presence or absence of certain training data may lead to different responses for the same inference request. This discrepancy can be exploited by the attackers. If the model's response to a current inference request remains consistent before and after the erasure of these data, it becomes feasible to postpone the unlearning execution. However, the prerequisite for verifying consistency is to execute unlearning to obtain the current response, which paradoxically contradicts our goal of deferring unlearning. This scenario is akin to needing a key to unlock a door, but the key itself is locked behind that very door.
To overcome this conundrum, we leverage the inherent structural characteristics of SISA to devise a certified inference consistency mechanism. This mechanism enables us to predict whether the response would change after unlearning execution without executing the unlearning. If deemed consistent, the unlearning execution can be postponed, allowing for an immediate response to the inference request. All deferred unlearning requests can then be batch-processed, significantly reducing processing time and computational overhead without introducing additional privacy risks.

Armed with the proposed new mechanism, we comprehensively investigate its adoption in MLaaS by identifying critical factors arising from the diversity in server specifications, the available computational resources, and the disparities in service objectives. Subsequently, we devise three core design options, including the ability to concurrently manage both training and inference environments, the timing of unlearning execution, and the strategy for handling inference requests that fail to achieve consistency. Based on these options, we have developed seven distinct instantiations of \projectname, allowing the server to opt for a suitable variant to achieve their desired balance between inference latency, resource consumption, and privacy protection. 

Our contributions are summarized below.
\begin{itemize}[leftmargin=*]
  \item To the best of our knowledge, we are the first to identify the security and privacy vulnerabilities of machine unlearning in MLaaS through the novel lens of the interplay between inference and unlearning requests, which are overlooked by existing machine unlearning literature taking the inference-oblivion approach.
  \item We propose \projectname, a new machine unlearning framework for MLaaS developed in an inference serving-aware manner, which can reduce inference latency while avoiding security and privacy vulnerabilities.
  \item We further propose seven instantiations under \projectname ~based on three key design options, each being tailored to specific environments and preferences of different MLaaS systems.
  \item We perform a comprehensive evaluation on four machine learning models using four real-world datasets to assess the performance of \projectname. 
  The results demonstrate that \projectname ~and its variants significantly reduce inference latency and computational overhead compared to the inference-oblivious baseline in different settings, achieving up to 99\% saving in inference latency and 31\% saving in computation overhead. 
\end{itemize}
Due to page limits, please refer to \url{https://arxiv.org/abs/2311.16136} for further discussions, detailed proofs and additional experiments.
\section{Background and Preliminary}
\label{sec.background}
\subsection{MLaaS}
Most existing MLaaS have two phases: the model training phase and the inference serving phase.  Since the two phases have different design objectives and resource preferences,  the training and inference contexts are usually optimized and maintained separately. 
In this paper, we focus on the general classification model in a fully supervised learning setting, which is fundamental and widely adopted in various applications. 

\partitle{ML model and training} ML model $F_{\bm{\theta}}:\mathcal{X} \to \mathbb{Y}$ with model parameters $\bm{\theta}$ maps from the feature space $\mathcal{X}$ to the discrete label set $\mathbb{Y}$ based on the training dataset $\tD$.  The training dataset $\tD = \{(\x_1,y_1),\dots,(\x_N,y_N)\}$ is collected from $N$ data owners, where for each $i\in \{1,\dots,N\}$, $\x_i \in \mathcal{X}$ and $y_i \in \mathbb{Y}$ are data owner $i$'s feature and label. $F_{\bm{\theta}}$ takes the empirical risk minimization form $F_{\bm{\theta}} = \frac{1}{N}\sum _{i=1}^N l_{\bm{\theta}}(\x_i,y_i)$, where $l_{\bm{\theta}}(\x_i,y_i)$ is the loss function on each individual $(\x_i,y_i)$. The objective of the ML training phase is to find the optimal $\bm{\theta}^*$ to minimize the overall loss on the training dataset, i.e., $\bm{\theta}^* := \arg\min_{\bm{\theta}} \{F_{\bm{\theta}} = \frac{1}{N}\sum _{i=1}^N l_{\bm{\theta}}(\x_i,y_i)\}$. The training phase is an offline process that can require significant computation, sometimes taking hours or even days to train a model from scratch.

\partitle{ML inference and inference serving}
Given the trained ML model, the ML inference takes an inference sample $\z \in\mathcal{X}$ and provides the inference result by $F(\z) = y\in \mathbb{Y}$ (for notational convenience, we will omit $\bm{\theta}$ from $F_{\bm{\theta}}(\z)$ in the sequel). The inference serving of MLaaS has the primary objective of delivering timely inference results to application users. The inference phase is an online/interactive process focused on achieving low inference latency. It is a challenging problem to balance response time with cost and resource effectiveness under dynamic workloads and fluctuating resource availability.
Many approaches have been proposed that can achieve promising frontiers among inference latency and cost/resources effectiveness \cite{DBLP:conf/nsdi/CrankshawWZFGS17,DBLP:conf/usenix/ZhangYWY19,DBLP:conf/usenix/Romero0YK21,gunasekaran2022cocktail}. 


\subsection{Machine Unlearning}

Machine unlearning is the task of removing a specific training sample from the trained ML model. The most basic yet inefficient approach is ``retraining-from-scratch'', which retrains the ML model on the remaining training set (i.e., the original training set excluding the sample to be unlearned) from scratch. To avoid the prohibitive computational cost of ``retraining-from-scratch'', many efficient machine unlearning mechanisms have been proposed, which can be roughly categorized into exact unlearning mechanisms \cite{cao2015towards, ginart2019making, brophy2021machine, bourtoule2021machine, ullah2021machine, schelter2021hedgecut, chen2022graph, chen2022recommendation,koch2023no,wang2023inductive,tao2024communication} and approximate unlearning mechanisms \cite{golatkar2020eternal, wu2020deltagrad, golatkar2020forgetting, nguyen2020variational, neel2021descent, peste2021ssse, golatkar2021mixed, warnecke2021machine, izzo2021approximate, mahadevan2021certifiable,liu2023muter,liu2024certified,lin2023erm}.  Exact machine unlearning mechanisms produce models identical to the one produced by ``retraining-from-scratch'', therefore capable of completely removing the requested data. Approximate unlearning mechanisms, on the other hand, trade exactness in data removal for improved computational and memory efficiency.

\partitle{SISA} 
In this paper, we focus on a notable exact unlearning mechanism called SISA \cite{bourtoule2021machine}, which is known for its ability to completely remove data and its generality to be applicable to various common ML models. During training, SISA divides the training dataset into multiple shards and builds a constituent model on each shard. Inference results are obtained by aggregating answers from all constituent models. During unlearning, only the shard containing the data to be unlearned is retrained to update the corresponding constituent model.
Unlike other exact unlearning algorithms, which are typically confined to convex learning scenarios \cite{cao2015towards, izzo2021approximate, mahadevan2021certifiable, nguyen2020variational, schelter2021hedgecut, warnecke2021machine} or specific learning algorithms such as random forests \cite{brophy2021machine} and K-means \cite{ginart2019making}, SISA can be applied across diverse machine learning contexts, including non-convex deep models \cite {chen2022graph, chen2022recommendation}.


\section{Problem Definition and Motivating Threats}
\label{sec.definition.and.threat.model}

In this section, we first introduce the threat model of machine unlearning in MLaaS in Sec.\ref{subsec.threat.model}. Then, we present two toy attacks that pose new security and privacy threats to machine unlearning in MLaaS in Sec.\ref{subsec.toy.threats}. These attacks motivate our design goals, which are discussed in Sec.\ref{subsec.design.goals}.

\subsection{Threat Model}
\label{subsec.threat.model}

There are three parties in machine unlearning in MLaaS: data owners, queriers, and the server. 

\partitle{Data owners and unlearning requests}
Data owners contribute their personal data to the server for the training of ML models. As granted by the ``right to be forgotten'' legislation, data owners have the right to request the removal of their data from the server to protect their privacy. After the ML model is trained and deployed in MLaaS, data owners can submit unlearning requests at any time. In response to an unlearning request, the server executes the machine unlearning mechanism to unlearn the relevant data from the trained model.
Some data owners can have malicious purposes such as attempting to increase the response latency for inference requests, rather than genuinely seeking to protect their privacy. 

\partitle{Queriers and inference requests}
Queriers are users of the application supported by MLaaS. After the ML model is trained and deployed, queriers can submit inference requests to the server at any time and receive the inference result from the server. 
The time it takes from submitting the request to receiving the results is known as the inference latency or response time.
Ideally, queriers want the inference latency to be as short as possible to minimize their wait time. 
If the inference latency exceeds an acceptable limit, queriers may abandon the service provider due to poor user experience and the delayed responses may cease to be useful \cite{hazelwood2018applied,gupta2020deeprecsys,gunasekaran2022cocktail}. Some queriers can have malicious purposes for submitting inference requests. Among the various widely recognized malicious inference threats, a prominent example is the privacy attack\cite{shokri2017membership,DBLP:conf/ndss/Salem0HBF019,li2021membership,liu2022ml,huang2022privacy}, in which queriers attempt to steal sensitive information of targeted data owners based on the inference results.

\partitle{Server}
In this paper, the server represents a service provider offering MLaaS. The server operates in two phases in existing MLaaS literature. During the training phase, the server builds the model using the training dataset that contains personal and potentially sensitive information from data owners. During the inference serving phase, the server processes queriers' inference requests by feeding the inference sample to the models and returning the result to the querier. To adhere to the RTBF regulations, the server further takes into account unlearning requests from data owners during the inference serving phase. The server also needs to erase newly identified poisoning data or unauthorized copyrighted data. 


\subsection{Two Motivating Threats on Machine Unlearning in MLaaS}
\label{subsec.toy.threats}
We demonstrate that existing inference-oblivion unlearning is susceptible to security and privacy vulnerabilities in the MLaaS scenario.
Specifically, we present two toy attacks that target two simplistic strategies for scheduling unlearning and inference requests. 




\partitle{Security threat: Service Obsolescence}
In the first strategy for processing unlearning requests, the server immediately executes the machine unlearning mechanism for every incoming unlearning request, a.k.a., ``unlearning-request-first'' strategy. Meanwhile, all inference requests submitted during the execution of unlearning have to be delayed until the server completes the update and deploys the new model for inference service once again. 

The toy attack targeting this strategy can infinitely block the processing of inference requests, ultimately preventing the server from providing any normal inference services.
The attack is carried out by adversarial data owners. Once the model is deployed, they periodically submit unlearning requests to the server to demand one of their data be unlearned. In an extreme and simplified scenario, they submit successive unlearning requests at intervals equal to the time it takes for the server to complete one machine unlearning process. As a result, all inference requests submitted at any time are blocked by these malicious unlearning requests. 
In fact, it suffices for the attack to cause the response latency to exceed an acceptable limit rather than infinitely blocking the inference requests.

\partitle{Privacy threat: Undesirable Exposure}
In the second strategy for processing unlearning requests, the server does not process the unlearning requests until a certain number of the unlearning requests have accumulated or a predefined waiting time has been reached. Meanwhile, the server continues processing incoming inference requests immediately based on the old model, a.k.a., ``inference-request-first'' strategy. It has the smallest response latency for inference requests submitted before the next unlearning execution for pending unlearning requests. 


The toy attack targeting the ``inference-request-first'' strategy could incur extra privacy risks for data owners who have submitted normal unlearning requests. 
The attack is carried out by adversarial queriers who attempt to infer privacy information of a targeted unlearning requester's data by submitting malicious inference requests to the server as in MIAs. 
Because the target's data is not timely unlearned from the model even after the deletion had been requested until the next unlearning execution (which already violates the RTBF principle), 
malicious inference requests processed during the suspention give the attacker an extra advantage in stealing the target data’s privacy.
If an unlearning request is intended to remove poisoning data, then the harm caused by the poisoning data will persist, and the potential ``landmines'' it poses can be triggered at any time during the suspension. If the unlearning request aims to delete copyrighted data, the losses from copyright infringement will continue to accumulate during the suspension. Every second of delay has the potential to incur additional losses.


\subsection{Desiderata for Unlearning in MLaaS}
\label{subsec.design.goals}
The above two motivating threats underscore the importance of developing an inference serving-aware machine unlearning framework in the MLaaS setting.
The server is responsible for cautiously scheduling the processing of inference and unlearning requests to minimize inference latency and extra risks. In detail, we have the following two design goals for the server.\\
  \underline{\emph{Low inference latency:}} For MLaaS admitting unlearning requests, the inference latency is increased by the unlearning process for two reasons: 1) Executing the machine unlearning mechanism takes much longer than inference; 2) The subsequent inference requests have to wait until the completion of unlearning update. Therefore, the server’s first design goal is to minimize the extra inference latency caused by unlearning update for as many inference requests as possible. 
  One way to reduce the inference latency is by strategically postponing the process of certain unlearning requests. This has two benefits: 1) subsequent inference requests can be processed without waiting for unlearning execution; 2) multiple unlearning requests can be batch-processed by executing a single machine unlearning update, decreasing the time spent on unlearning. \\
  \underline{\emph{Avoid extra risk:}} Achieving the first design goal may increase the risk for the unlearning requesters. 
  If an unlearning request is postponed, the requested data remains in the model and may be utilized by the attacker through inferences made on the model.
  To prevent this, the server’s second design goal is to ensure that the postponed unlearning execution does not increase the risk for the requesters.

\partitle{What we do not achieve in this paper}
There are also other opportunities for improvement in both the security and efficiency of the machine unlearning scheme. However, these areas are orthogonal to the research presented in this paper.\\
\underline{\emph{Regarding extra privacy risk:}} We focus on ensuring that data owners who have submitted unlearning requests do not suffer from extra privacy risk due to the postponed unlearning execution. However, we do not address other sources of privacy threats, such as general privacy attacks considered in MIA literature where any data owners can be the victim \cite{shokri2017membership,DBLP:conf/ndss/Salem0HBF019,li2021membership,liu2022ml}.

\underline{\emph{Regarding extra inference latency:}} We focus on reducing the extra inference latency caused by unlearning requests. However, we do not address other sources of inference latency in the original inference serving pipeline that can be raised from various external factors (e.g., network issues, busty workload) or internal factors (availability of hardware resources changes, concurrency from other processes). In addition, we do not seek to accelerate the machine unlearning mechanism itself, as this has been studied in many existing literature and is not the focus of this paper.
  
\section{\projectname: Machine Unlearning in MLaaS}
\label{sec.framework}

Motivated by the security and privacy vulnerabilities and design goals, we propose \projectname, a new framework that takes an inference serving-aware approach to machine unlearning in MLaaS. 

\subsection{Intuition and Overview}
\label{subsec.intuition}
\partitle{Intuition}
As mentioned in the first design goal, we can reduce inference latency by strategically postponing the execution of certain unlearning requests and processing subsequent inference requests first. 
Despite the benefit of reducing latency, we need to ensure that attackers do not gain extra advantage from the postponement. At first glance, the key to the server's design is selecting the right unlearning requests to postpone. Rather, we can take a complementary (but conceptually equivalent) view by selecting from subsequent inference requests that can be answered before executing the pending unlearning requests. The basic principle is to select inferences having consistent inference results with or without unlearning execution of the pending unlearning requests. This way, a black-box privacy attacker cannot gain any extra advantage based on the inference results. Otherwise, if the inference result is inconsistent, which means the result can be significantly influenced by the yet-to-be-unlearned sensitive data, it cannot be sent to the querier without first processing the unlearning request.

With this selection principle, the challenge is how to certify that the inference result is consistent on models with and without unlearning execution of the pending unlearning requests. Apparently, it is infeasible to compare the inference results by actually executing unlearning, as this would defeat our goal of avoiding its execution in the first place. To certify inference consistency without actual machine unlearning execution, we resort to the concept of model robustness with respect to small changes in the training dataset \cite{DBLP:conf/iclr/0001F21,wang2022improved,jia2021intrinsic}. Roughly speaking, this refers to the property that two ML models trained on two datasets which differ only on a few data have the same inference results. 
This naturally evokes the idea of utilizing Differential Privacy (DP) \cite{dwork2006calibrating} techniques for model training. However, the primary challenge with the DP-based approach is that the robustness offered by DP allows only for approximate unlearning, rather than exact unlearning. As the number of unlearning request keeps increasing during the MLaaS service, the privacy bound of DP may be breached with the growth of to-be-erased data. Additionally, the noise introduced by DP can significantly decrease the model's utility.
The ideal robustness entails two favoring consequences: 1) more subsequent inference requests can be responded to before processing the pending unlearning requests, reducing response latency; 2) more unlearning requests can be accumulated and executed all at once by a single unlearning execution, reducing unlearning execution time. 

Finally, we arrive at the two concrete questions to be answered: 1) How to convert any ML models to a more robust counterpart in a generally applicable way that is also compatible with one of the existing machine unlearning mechanisms; 2) How to connect the robustness concept with the inference consistency to certify that certain inference requests can be responded without actual unlearning execution for the pending unlearning requests. 

\begin{algorithm}[htbp]
    \caption{\projectname: training and inference with certified inference consistency check}
    \label{alg: framework}
    {\small
        {
            \begin{algorithmic}[1]
                \Require {Training dataset $\tD = \{\x_1,\dots,\x_N\}$}
                \renewcommand{\algorithmicrequire}{ \textbf{Training:}}
                \Require 
                    \State {Randomly divides $\tD$ into $K$ shards $\tS_1,\dots,\tS_K$}
                        \State{Train constituent model $f_k$ on $\tS_k$, for all $k \in [1,2, \dots, K]$}
                    \State {Launch inference service with $F=\{f_1,\dots,f_K\}$ }
                \renewcommand{\algorithmicrequire}{ \textbf{Inference and Unlearning Serving:}}
                \Require
                    \For {$t\in [1, 2, \dots]$}
                            \If {Receive an Unlearning Request $\ur = (i,'Unlearn')$}
                                \State {Record it as a pending unlearning request}
                            \ElsIf {Receive an Inference Request $\ur = (\z,'Inference')$}
                                \State Inference with constituent models: {${f_{1}^{\tt O}(\z), \dots,f_{K}^{\tt O}(\z)}$}
                                \State {Count: ${\tt Count}_{y}^{\tt O}(\z) := |\{k\in[K]\Big{|}f_k^{\tt O}(\z) = y \}|.$}
                                \If {satisfy certified inference consistency condition }
                                    \State{Return $F(\z) := \arg\max _{y\in \mathbb{Y}} {\tt Count}_{y}^{\tt O}(\z)$.}
                                \Else
                                    \State{Process in accordance with specific design options.}
                                \EndIf
                            \EndIf
                    \EndFor
                \Ensure {Responses for inference and unlearning requests.}
            \end{algorithmic}
        }
    }
\end{algorithm}

\partitle{Overview of \projectname}
Taking on the above intuition, we propose the \projectname ~framework with two phases, i.e., the training phase and the inference\&unlearning serving phase, where the latter supports both inference and unlearning requests. 

During the training phase, \projectname ~is given the training dataset, the model architecture and the loss designated by the application. Then, \projectname ~employs the same training strategy as SISA, dividing the training dataset into $K$ shards and building each constituent model independently on its shard. As will be shown in Sec.\ref{sec.framework}, this shard-aggregate strategy makes the model more robust to small changes in the training dataset (i.e., the small amount of unlearned data in our context), in addition to its known unlearning efficiency. 
In this sense, \projectname ~uncovers a hidden benefit of SISA when applied to machine unlearning in MLaaS. 

During the serving phase, \projectname ~aims to avoid frequent interruptions of normal inference service for processing unlearning requests. 
For an incoming unlearning request, \projectname ~records it and the shard it belongs to but does not necessarily execute the machine unlearning immediately (i.e., treats it as a pending unlearning request). \projectname ~will keep the data and shard index for all pending unlearning requests until they are unlearned. 
For an incoming inference query, \projectname ~derives a certified inference consistency bound to determine if the inference result will be consistent with or without processing all pending unlearning requests.
If the inference has consistent results, \projectname ~immediately returns the inference result to the querier. The computation of the certified inference consistency mechanism essentially involves processing the inference request itself plus some simple counting and comparing operations, with no actual machine unlearning execution. 
\projectname ~thus incurs minimal additional inference latency caused by unlearning. Otherwise, if the certified inference consistency condition is not met, \projectname ~may need to halt the inference service and execute the unlearning first. However, due to the shard-aggregate strategy and increased robustness, most of the inference requests have certified inference consistency and can be immediately responded to. 

Algorithm \ref{alg: framework} summarizes the training and inference with inference consistency certification in \projectname. The unlearning execution is omitted as it depends on design options.

In Sec.\ref{subsec.formalization}, we provide a formalization of \projectname. In Sec.\ref{subsec.certified.inference.consistency}, we derive the certified inference consistency mechanism. In Sec.\ref{subsec.analysis}, we analyze \projectname ~in terms of inference latency to show its advantage compared to the inference-oblivion baseline approach. 

\begin{figure*}[htbp]
	\centering
	\includegraphics[width=0.96\textwidth]{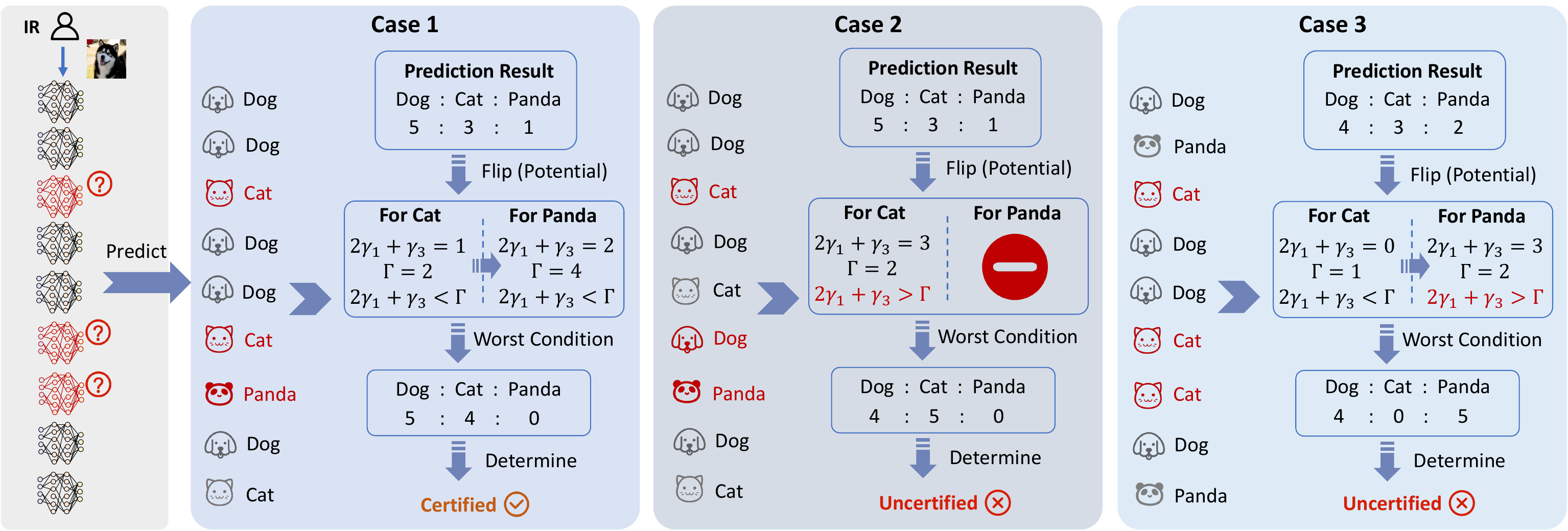}
	\caption{Illustration of three representative cases for the certified inference consistency mechanism.}
        \Description{This figure illustrates three representative cases for the certified inference consistency mechanism in a machine learning model. The left part of the figure shows an inference request (IR) for an image of a husky, which is processed by multiple models, each represented by neural network icons. Some models are marked with red question marks, indicating uncertainty in their predictions. Case 1: The prediction result is Dog: 5, Cat: 3, Panda: 1. A potential flip is considered, and for the Cat and Panda classes, the inequalities 2y_1 + y_3 < \Gamma are checked with different values of \Gamma. The worst condition results in Dog: 5, Cat: 4, Panda: 0. The case is determined to be certified. Case 2: The prediction result is Dog: 5, Cat: 3, Panda: 1. A potential flip is considered, and the inequality 2y_1 + y_3 > \Gamma is checked for the Cat class, resulting in Dog: 4, Cat: 5, Panda: 0 under the worst condition. This case is determined to be uncertified. Case 3: The prediction result is Dog: 4, Cat: 3, Panda: 2. A potential flip is considered, and the inequalities for the Cat and Panda classes are checked, resulting in Dog: 4, Cat: 0, Panda: 5 under the worst condition. This case is also determined to be uncertified. Each case involves checking conditions to determine if the prediction remains consistent under potential flips, with results categorized as certified or uncertified based on the consistency check outcomes.}
	\label{fig: certification}
\end{figure*}

\subsection{Formalization of \projectname}
\label{subsec.formalization}
Based on the overview, we formalize the key functions (i.e., training, inference, and unlearning) in the two phases of \projectname ~below.

\partitle{Training}
To facilitate unlearning, \projectname ~builds the model $F$ by shard-and-aggregate strategy following SISA. First, it randomly divides the training dataset $\tD$ into $K$ shards $\tS_1,\dots,\tS_K$. On each shard $\tS_k$, \projectname ~trains a constituent model $f_k$, which will result in $K$ constituent models $\{f_1,\dots,f_K\}$.

\partitle{Inference}
Denote an inference request by $\qr = (\z,'Inference')$. Each constituent model provides prediction class $y_k = f_{k}(\z)$ on the inference request sample $\z$. To provide the final inference result $F(\z)$, \projectname ~aggregates the prediction results from all $K$ constituent models by majority voting. That is, \projectname ~counts the number of constituent models that has inference result $y$ for each candidate label $y\in \mathbb{Y}$, given by
\begin{equation}
  {\tt Count}_{y}(\z) \eqdef |\{k\in[K]\Big{|}f_k(\z) = y \}|.
\end{equation}
The final inference result $F(\z)$ is the label that has the largest count,
\begin{equation}
  F(\z) \eqdef \arg\max _{y\in \mathbb{Y}} {\tt Count}_{y}(\z),\label{eq. voting}
\end{equation}
where ties are broken by returning the label with a smaller index. 

\partitle{Unlearning}
Denote an unlearning request $\ur = (i,'Unlearn')$ is sent by the ${\tt i}$-th data owner to request for unlearning his/her data $(\x_i,y_i)$ from the trained ML model. The server keeps track of all pending unlearning requests. Let $t^{\tt O}$ be the timestamp of the most recent training/unlearning execution. For any shard $k\in [K]$, $\tS_k^{\tt O}$ be the remaining data partition of the $k$-th shard at $t^{\tt O}$, $f^{\tt O}_k$ be the model obtained by unlearning execution on $\tS_k^{\tt O}$. 
In other words, $\tS_k^{\tt O}$ has no pending unlearning requests at $t^{\tt O}$. 
Then, from $t^{\tt O}$ to the current timestamp $t$, we denote the set of pending unlearning requests (if any) for the $k$-th shard during this period by $\UR_k^t\eqdef \{\ur_k^{r_1},\dots, \ur_k^{r_m}\}$, and thus $\UR^t \eqdef (\UR_1^t,\dots,\UR_K^t)$. Let $\tS_k^{t} \eqdef \tS_k^{\tt O} \circ \ur_k^{r_1} \dots \circ \ur_k^{r_m}$ be the $k$-th shard with all pending unlearning requests $ \UR_k^t$ removed from $\tS_k^{\tt O}$. In addition, we define $f_k^{t}$ to be the constituent model after unlearning at $t$ on $\tS_k^{t}$ to process $\UR_k^t$. 
Collecting all $K$ shards, we have $\tD^{\tt O} = \tS_1^{\tt O}  \cup \tS_2^{\tt O}  \cup \dots \cup \tS_K^{\tt O}$, $\tD^{t} = \tS_1^{t} \cup \tS_2^{t} \cup \dots \cup \tS_K^{t}$.
Correspondingly, we introduce the following notations,
\begin{gather}
    {\tt Count}_{y}^{\tt O}(\z) \eqdef |\{k\in[K]\Big{|}f_k^{\tt O}(\z) = y \}|, \\
    {\tt Count}_{y}^t(\z) \eqdef |\{k\in[K]\Big{|}f_k^t(\z) = y \}|;\\
    F^{\tt O}(\z) \eqdef \arg\max _{y\in \mathbb{Y}} {\tt Count}_{y}^{\tt O}(\z), \label{eq.without.unlearning}\\
    F^t(\z) \eqdef \arg\max _{y\in \mathbb{Y}} {\tt Count}_{y}^t(\z)\label{eq.with.unlearning}.
\end{gather}
From $t^{\tt O}$ to $t$, \projectname ~process inference requests based on the old model $F^{\tt O}$, which calls for certified inference consistency to avoid extra privacy risks, as presented below.










\subsection{Certified Inference Consistency}
\label{subsec.certified.inference.consistency}
\partitle{Definition of inference consistency}
According to the threat model, the adversary can only access the final inference results (i.e., $F^{\tt O}(\z)$). Thus, the postponed unlearning will not cause extra privacy risks if $F^{\tt O}(\z) = F^t(\z)$, i.e., the inference results on $\z$ are consistent between the model with and without timely unlearning. Under this inference consistency condition, the inference request on $\z$ can be immediately responded to the queirer without the actual unlearning to obtain $F^t$ at timestamp $t$. This is formalized below.

\begin{definition}[\textbf{Inference Consistency}]
\label{def.inference.consistency}
For the inference request on sample $\z$, the MLaaS has inference consistency between the ML models with and without processing the pending unlearning requests, if they have the same prediction results on $\z$. That is, it has $F^{\tt O}(\z) = F^t(\z)$, where $F^{\tt O}$ is the model without $\UR^t$ unlearned and $F^t$ is the model with actual unlearning, as specified in eq.(\ref{eq.without.unlearning}) and eq.(\ref{eq.with.unlearning}), respectively.
\end{definition}
\partitle{Certified inference consistency}
It is infeasible in practice to exactly check the condition in Definition \ref{def.inference.consistency}, because one has to execute the machine unlearning mechanism to obtain $f_{1}^t,\dots,f_{K}^t$ in the first place. Rather, we need to certify the inference consistency based on the prediction results made by $f_{1}^{\tt O},\dots,f_{K}^{\tt O}$ and the records of pending unlearning requests. 
This could be accomplished by assessing whether the outcome of the majority voting (eq. \ref{eq. voting}) would change if the prediction results from the constituent models that are potentially impacted by $\UR^t$ (denoted red in Figure \ref{fig: certification}) were to flip to a label different from $F^{\tt O}(\z)$ after unlearning execution. Inference consistency is maintained if the result remains unchanged, even under the worst-case scenario.
We derive Theorem \ref{thm.certified.inference.consistency} to provide a 
formalized certification condition.
Consequently, as long as the condition in Theorem \ref{thm.certified.inference.consistency} holds, we can ensure the certified consistency without the actual machine unlearning execution.



\begin{theorem}[\textbf{Certified Inference Consistency}] 
\label{thm.certified.inference.consistency}
Suppose that the most recent unlearning-updated models are trained on $\tD ^{\tt O}$, the pending unlearning requests are $\UR^t$, and the shards with pending unlearning requests are $\{k\in[K] \big{|} \tS_k^t \neq \tS_k^{\tt O}\}$. Given inference sample $\z$, let $y_a$ be the labels having maximum counts from $f_{1}^{\tt O}(\z),\dots,f_{K}^{\tt O}(\z)$ and $y_b$ is a label different from $y_a$, i.e., $\forall y_b\in\mathbb{Y}, y_b \neq y_a$. Let the counts $\gamma _1, \gamma_2, \gamma_3$ be defined as,
\begin{align}
    \gamma_1\eqdef &|k\in[K]\big{|} \tS_k^t \neq \tS_k^{\tt O} \wedge f_k^0(\z) = y_a|;\\
    \gamma_2\eqdef &|k\in[K]\big{|} \tS_k^t \neq \tS_k^{\tt O} \wedge f_k^0(\z) = y_b|;\\
    \gamma_3\eqdef &|k\in[K]\big{|} \tS_k^t \neq \tS_k^{\tt O} \wedge f_k^0(\z) \neq \{y_a, y_b\}|.
\end{align}
Then, \projectname ~has certified inference consistency for the inference request on $\z$, if the following condition holds: $\forall y_b\in\mathbb{Y}, y_b\neq y_a$,
\begin{equation}
    2\gamma_1 + \gamma _3 \leq \Gamma_b,
\end{equation}
where $\Gamma_b = {\tt Count}_{y_a}^{\tt O}(\z) - {\tt Count}_{y_b}^{\tt O}(\z) - \mathbb{I}(y_b<y_a)$ with $\mathbb{I}(\cdot)$ being the indicator function. In other words, the inference request on $\z$ can be responded to immediately without incurring extra privacy risk while avoiding waiting for the unlearning update to obtain $f_{1}^{t}(\z),\dots,f_{K}^{t}(\z)$. 
\end{theorem}

\label{appendix. consistency}
\begin{algorithm}[htbp]
    \caption{Inference Consistency Certification Mechanism}
    \label{alg: consistency}
    {\small
        {
            \begin{algorithmic}[1]
                \Require {Inference request $\qr = (\z,'Inference')$, models after last unlearning update $F^{\tt O}\eqdef\{f_1^{\tt O}, \cdots, f_k^{\tt O}\}$, dataset after last unlearning update $\tD^{\tt O} = \tS_1^{\tt O}  \cup \tS_2^{\tt O}  \cup \dots \cup \tS_K^{\tt O}$, unlearning requests submitted since last unlearning update: $\UR^t \eqdef (\UR_1^t,\dots,\UR_K^t)$}.
                    \State{$certification_{\qr}=True$}
                    \For {$k \in [1,2, \dots, K]$}
                        \State{Inference with constituent model: $f_k^{\tt O}(\z)$}
                        \If {$\tS_k^{\tt O}\cap \UR_k^t \neq \varnothing$ }
                            \State {Mark as $\tS_k^t \neq \tS_k^{\tt O}$}
                        \EndIf
                    \EndFor
                    \State {$y_a=\arg\max _{y\in \mathbb{Y}} |\{k\in[K]\Big{|}f_k^{\tt O}(\z) = y \}|$}
                    \For {$\forall y_b\in\mathbb{Y}, y_b \neq y_a$}
                        \State{$\gamma_1\eqdef |k\in[K]\big{|} \tS_k^t \neq \tS_k^{\tt O} \wedge f_k^0(\z) = y_a|$}
                        \State{$\gamma_2\eqdef |k\in[K]\big{|} \tS_k^t \neq \tS_k^{\tt O} \wedge f_k^0(\z) = y_b|$}
                        \State{$\gamma_3\eqdef |k\in[K]\big{|} \tS_k^t \neq \tS_k^{\tt O} \wedge f_k^0(\z) \neq \{y_a, y_b\}|$}
                        \State{$\Gamma_b = {\tt Count}_{y_a}^{\tt O}(\z) - {\tt Count}_{y_b}^{\tt O}(\z) - \mathbb{I}(y_b<y_a)$}
                        \If{$2\gamma_1 + \gamma _3 > \Gamma_b$}
                            \State{$certification_{\qr}=False$}
                            \State{Break}
                        \EndIf
                    \EndFor
                    \If{$certification_{\qr}$}
                        \State {$\qr$ has certified consistency and can be responded to immediately.}
                    \Else
                        \State{The consistency of $\qr$ can not be certified}
                    \EndIf
                \Ensure {$certification_{\qr}$}
            \end{algorithmic}
        }
    }
\end{algorithm}

All the proofs in this paper is relegated to Appendices in the arXiv version. Figure \ref{fig: certification} provides an illustration of different cases within the mechanism. Algorithm \ref{alg: consistency} summarises the inference consistency certification mechanism.

\subsection{Extension to Different Unlearning Methods}
Although the design of \projectname~ draws inspiration from SISA, which retrains the constituent models from scratch, \projectname~ can also be extended beyond SISA. Any exact-unlearning algorithm applicable to the target model can be utilized within the ERASER scheme as a substitute for “retraining from scratch” to update constituent models. For $(\epsilon, \delta)$-approximate-unlearning algorithms, the accumulation of multiple unlearning requests might continuously amplify $\epsilon$ and $\delta$. However, we believe that leveraging composition techniques from differential privacy could help address this issue, which will be further explored in our future work.

\section{Design Options and Variants of \projectname}
\label{sec.variants}

The substantial variations in service objectives and available resources across different MLaaS pose significant challenges for MLaaS providers to choose a suitable strategy to achieve their desired balance between inference latency, resource consumption, and privacy protection. Consequently, we have investigated existing MLaaS solutions and identified three key factors that significantly impact the performance of \projectname ~as design options for system deployment.

\subsection{Design Options}
\label{subsec.options}

\partitle{Design option I: maintain single or double contexts}
We consider whether to have both inference and unlearning contexts or not. In existing MLaaS literature, it is widely recognized that the training and inference contexts have distinct characteristics. For instance, the training context involves loading the training dataset and is typically an offline computation, while inference is interactive and requires low latency. Thus, these two contexts are often studied and optimized separately. However, the emerging need for machine unlearning blurs the boundary between training and inference. Machine unlearning mechanisms share characteristics with both training and inference: on one hand, they require (partial) training data and more computation than inference, resembling training; on the other hand, they are interactive and require careful scheduling to avoid additional inference latency.
\begin{itemize}[leftmargin=*]
  \item \underline{\emph{Design Option I-A:}} \projectname ~opts to maintain both the inference context at the front end and the unlearning context at the back end, with dedicated resources at both ends that can operate in parallel. The former utilizes all older copies of the constituent models to provide inference service with certified inference consistency. 
  The latter executes the unlearning mechanism to update the corresponding constituent models that have pending unlearning requests to process.
  \item \underline{\emph{Design Option I-B:}} \projectname ~opts to maintain only one context at a time and switches between contexts when processing different types of requests. When switching to the unlearning context, all inference requests are suspended until switching back.
\end{itemize}

\partitle{Design Option II: Timing of unlearning}
We explore when to execute the pending unlearning requests and provide three options below depending on whether the unlearning requests are processed immediately or wait until there are inference requests that cannot satisfy the prediction consistency condition.
\begin{itemize}[leftmargin=*]
  \item \underline{\emph{Design Option II-A:}} \projectname ~opts to accumulate unlearning requests until an inference request fails to meet the inference consistency condition. At that point, all pending unlearning requests are batch-unlearned. Subsequent inference requests submitted during unlearning can only be responded to if they meet the inference consistency condition based on available model copies. This option is referred to as \emph{uncertification-triggered unlearning}.
  \item \underline{\emph{Design Option II-B:}} \projectname ~opts to immediately execute unlearning at the back end, while serving the inference requests with the old copies of the models in the front end. The inferences without consistency need to wait for the completion of ongoing unlearning execution. This option is named \emph{immediate unlearning}.
  \item \underline{\emph{Design Option II-C:}} When combined with \emph{Option III-C} introduced below, \projectname ~opts to execute unlearning until the number or ratio of inconsistent inference requests reaches a threshold. This option is referred to as \emph{threshold-triggered unlearning}.
\end{itemize}

These options all have their advantages and disadvantages. Option II-A\&C has a better chance of accumulating more pending unlearning requests and processing them together. However, since the unlearning update is triggered by inconsistent inference request, it's more likely to occur during periods of heavy inference request workload, potentially impacting more users' experience. Option II-B can significantly reduce the inference latency caused by waiting for unlearning execution. However, it may result in high computational overhead. In future work, we will explore more sophisticated design options for unlearning timing that take multiple factors, such as workload, resource utilization, user satisfaction, and service cost, into holistic consideration.

\partitle{Design Option III: Handling of uncertified inference requests}
We consider how to handle uncertified inference requests under the setting of \emph{Design Option II-C}. In conventional MLaaS that only deals with inference requests, the system design and computational resources are typically capable of handling the majority of inference requests and meeting the service-level objective (SLO) of response time. However, a small portion of them, known as tail inference requests \cite{DBLP:conf/osdi/GujaratiKAHKVM20,DBLP:conf/sc/Cui0CZLZSMYLG21,DBLP:conf/usenix/Cui00WLZ0G22}, can be difficult to meet the SLO due to various factors raised in the practical serving environment. There are several options to handle tail inference requests. One is to consider scaling up with more computational resources at the expense of higher costs. Another is to discard a small portion of them (based on a predetermined discarding ratio)  when the server finds that either the inference latency will exceed an acceptable limit or the cost for additional resources is too high.


As observed in our empirical results in Sec.\ref{sec.experiment}, a new factor in \projectname ~will cause the phenomenon of tail inference requests, i.e., failure to certify the inference consistency. This means that while most inference requests can certify inference consistency and be immediately responded to, a small portion cannot and may experience much longer inference latency as they may need to wait for the machine unlearning execution.
To handle such uncertified inference requests, we propose the third design option below.
\begin{itemize}[leftmargin=*]
  \item \underline{\emph{Design Option III-A:}} \projectname ~opts to immediately provide the uncertified inference results to the queriers as long as such responses are within a predetermined small ratio (uncertification ratio), introduced in the same spirit as the discarding ratio. 
  \item \underline{\emph{Design Option III-B:}} \projectname ~opts to store the uncertified inference requests and reprocess them with the new model after the next machine unlearning execution. 
\end{itemize}
For Option III-A, the advantage is that the inference latency for these inference requests can be significantly reduced. However, these uncertified inference results may introduce additional privacy risks. 
As such, we suggest adopting this option with more caution, e.g., by setting a very small uncertification ratio to limit the risk, or discarding such requests and not returning the results.


\subsection{Variants of \projectname}
\label{subsec.variants}

SISA with unlearning-request-first strategy, as described in Sec.\ref{subsec.toy.threats}, can be viewed as a variant with Options B-B-B. We propose seven more practical variants with different combinations of choices for three options. The details are shown in Table \ref{table: variants}, and a comparison among four double-context variants is presented in Figure \ref{fig: algorithms}.

\partitle{Option A-B-B: Double Contexts for Immediate Unlearning with Postponed Certified Inference (DIMP)}
DIMP maintains both the inference context and the unlearning context, choosing to execute unlearning requests immediately at the backend. Inference requests with certified results will receive immediate responses, while uncertified requests will be postponed until the ongoing unlearning update is completed. This variant offers low inference latency at the cost of high computation overhead since each unlearning execution can process only one unlearning request.

\partitle{Option B-A-B: Single Context for Uncertification-Triggered Unlearning with Postponed Certified Inference (SUTP)}
SUTP maintains a single context at a time and accumulates unlearning requests until an inference request fails to satisfy the certified inference consistency condition. Instead of updating immediately, constituent models corresponding to unlearning requests are placed on an unlearning waiting list. Once uncertification occurs, the server switches to the unlearning context and executes the unlearning update. Uncertified prediction requests will be postponed until the server completes the unlearning update and switches back to the inference context, as well as the inference requests that arrive during the unlearning. This variant saves computation overhead compared to DIMP, but at the cost of increased inference latency.

\partitle{Option A-A-B: Double Contexts for Uncertification-Triggered Unlearning with Postponed Certified Inference(DUTP)}
DUTP differs from SUTP by maintaining both the inference and unlearning contexts. In this variant, inference requests that arrive during unlearning execution receive immediate responses if the inference result is certified. Only uncertified inference requests are postponed. DUTP reduces inference latency compared to SUTP while consuming more resources by maintaining double contexts simultaneously.

\partitle{Option B-C-A: Single Context for Threshold-Triggered Unlearning with Uncertified Inference (STTU)}
STTU maintains a single context at a time and accumulates unlearning requests until the ratio of inference requests that fail to satisfy the consistency condition exceeds a threshold. Uncertified inference results within the predetermined threshold are provided immediately. Once the ratio exceeds the threshold, the server switches to the unlearning context and executes the unlearning update. The request that triggers the update and newly arrived inference requests during the update must wait for the context switch. This variant reduces the frequency of unlearning updates, decreasing computation overhead. However, a predetermined small ratio of inference requests will receive uncertified prediction results with potential privacy risks.

\partitle{Option A-C-A: Double Context for Threshold-Triggered Unlearning with Uncertified Inference (DTTU)}
DTTU differs from STTU by maintaining both the inference and unlearning contexts, reducing inference latency by responding to certified inference requests during unlearning updates.

\partitle{Option B-C-B: Single Context for Threshold-Triggered Unlearning with Postponed Certified Inference (STTP)}
STTP shares the same choices as STTU for Option I and Option II, but suspends uncertified inference requests in an inference waiting list rather than providing uncertified predictions immediately. These requests are postponed and re-processed after the next unlearning update. STTP guarantees that every inference will receive a certified response. The inference waiting list is then cleared and the uncertification ratio is reset to zero, making it take longer to trigger the next unlearning update. STTP saves more computation overhead than STTU at the cost of higher inference latency since uncertified inference requests need waiting for reprocessing.

\partitle{Option A-C-B: Double Context for Threshold-Triggered Unlearning with Postponed Certified Inference (DTTP)}
DTTP differs from STTP by maintaining double contexts, and reduces latency by serving inference requests during unlearning updates.

\begin{table}[htbp]
\centering
\footnotesize
    \caption{Summary of Variants under Different Design Options}
    \label{table: variants}
     \begin{center}
    \begin{tabular}{ c|ccc}
       \toprule
     \diagbox{\textbf{Variants}}{\textbf{Options}} &  \makebox[0.06\textwidth][c]{\textbf{I}} & \makebox[0.06\textwidth][c]{\textbf{II}} &\makebox[0.06\textwidth][c]{\textbf{III}} \\

       \midrule
      DIMP & A & B & B  \\
      SUTP & B & A & B  \\
      DUTP & A & A & B  \\
      STTU & B & C & A \\
      DTTU & A & C & A \\
      STTP & B & C & B  \\
      DTTP & A & C & B  \\
      \bottomrule
    \end{tabular}
       \end{center}
\end{table}


\begin{figure}[htbp]
	\centering
	\includegraphics[width=0.47\textwidth]{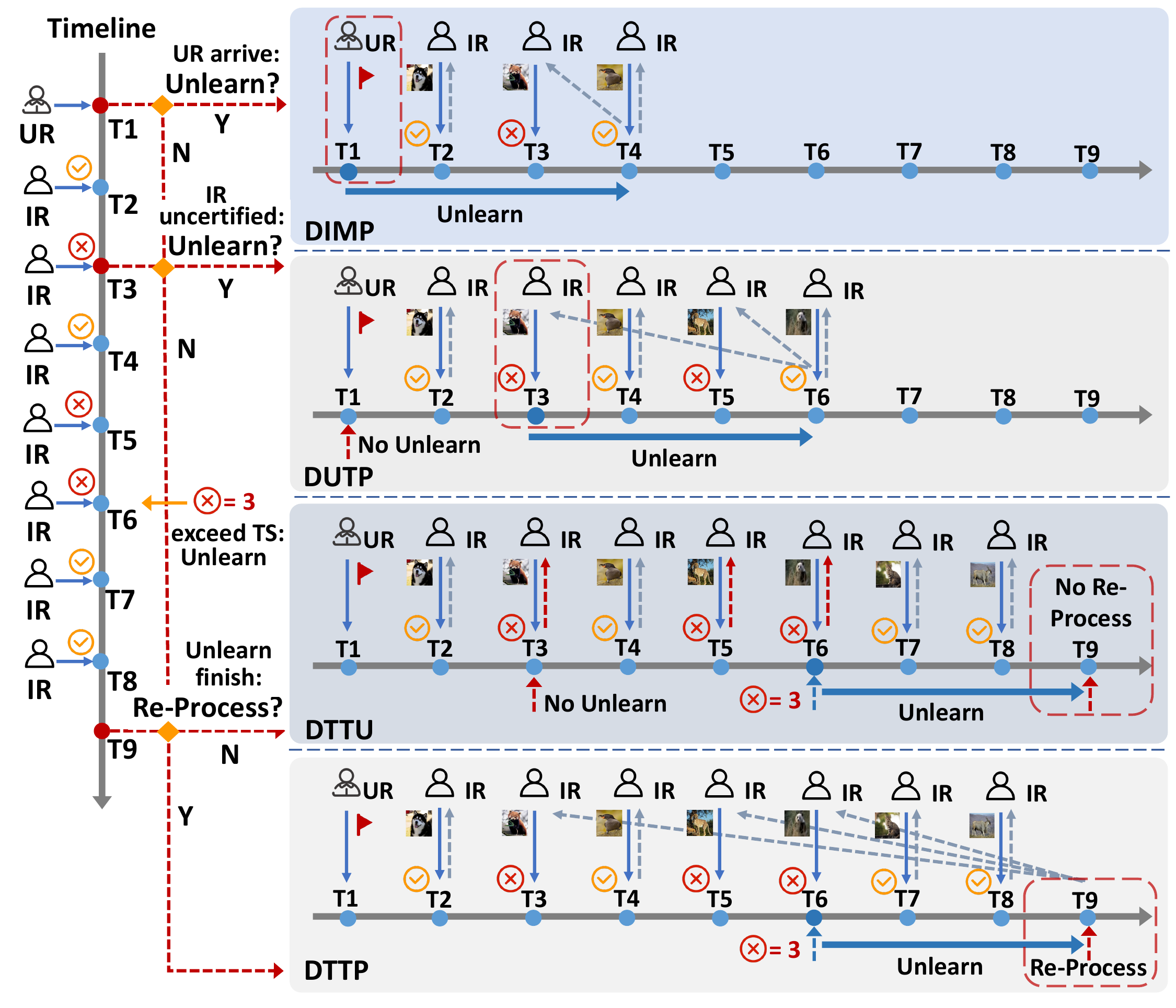}
	\caption{Illustration of variants with double contexts of \projectname.}
	\label{fig: algorithms}
        \Description{This figure illustrates four variants of the ERASER framework with double contexts for handling inference requests (IR) and unlearning requests (UR) over a timeline from T1 to T9. DIMP : At T1, a UR arrives and triggers the unlearning process. IRs continue to arrive and are processed with potential delays due to the unlearning process. Unlearning completes by T5. DUTP: At T1, a UR arrives, but unlearning is delayed until T5. Inference requests are processed without unlearning until T5. Unlearning begins at T5 and finishes by T7. DTTU: Inference requests are processed, and unlearning is triggered only when the number of uncertified IRs exceeds a threshold (set to 3) at T6. Unlearning starts at T6 and completes by T8. No re-processing of previous IRs occurs. DTTP: Similar to DTTU, but once unlearning finishes at T8, there is a re-processing phase for previous IRs that were processed before unlearning started. This re-processing occurs at T9. The timeline on the left shows the decision points for unlearning based on the arrival of URs and the certification status of IRs. The process flow within each variant depicts how IRs and URs are managed to maintain model consistency and privacy.}
\end{figure}

These different variants achieve a tradeoff between privacy, resource consumption, and inference latency. All double-context variants consume more resources to reduce inference latency. DIMP has the lowest inference latency among all variants, but its computation overhead is the highest. The computation overhead of DTTP is the lowest, but its inference latency is the highest. The inference latency and computation overhead of DUTP are between the first two variants. The inference latency of DTTU is close to that of DUTP, but the computation overhead is lower, at the cost of introducing privacy risk. Ultimately, the choice of variant depends on the specific requirements and constraints of a given application.
A server can also dynamically adjust among different variants based on the current computational resources and response time to pursue higher resource utilization.


\subsection{Analysis}
\label{subsec.analysis}
We take DIMP (Options A-B-B) and SISA with unlearning-request-first strategy (Options B-B-B)
as examples to theoretically analyze the acceleration capability of \projectname.

Assume the time required to retrain a constituent model is $r$, the number of unlearning requests is $n_u$, the number of inference requests is $n_i$, and the arrival time interval for all requests is $[0, T]$.
For simplicity, we assume that all unlearning requests arrive at fixed time intervals, with an unlearning request arriving every $\frac{T}{n_u}$ starting from $t=0$.
Since the time taken for a single inference is almost negligible compared to unlearning execution, we do not consider the time taken for inference in calculating. 

\begin{theorem} [Waiting Time of SISA]
    \label{theorem:sisa}
    The expected waiting time of each inference request in SISA is
    \begin{equation}
        \expect(w_{sisa})=
        \begin{cases}
            \frac{n_ur^2}{2T}, &\quad r\le \frac{T}{n_u}\\
            r-\frac{T}{2n_u}, &\quad r>\frac{T}{n_u}.
        \end{cases}
    \end{equation}
\end{theorem}
Theorem \ref{theorem:sisa} shows that $\expect(w_{sisa})$ is positively correlated with the frequency of unlearning requests $\frac{n_u}{T}$ and the retraining time $r$. 
\begin{theorem} [Waiting Time of DIMP]
    \label{theorem:DIMP}
    Assuming that the probability of the prediction result being uncertified at each judgment is $p_{uc}$, the upper bound of the expected inference request waiting time in DIMP is:
    \begin{equation}
        \expect(w_{dimp})\le
        \begin{cases}
            p_{uc}\cdot \frac{n_ur^2}{2T}, &\quad r\le \frac{T}{n_u}\\
            p_{uc}\cdot (r-\frac{T}{2n_u}), &\quad r>\frac{T}{n_u}
        \end{cases}
    \end{equation}
\end{theorem}
Given the complexity of the actual expression for $\expect(w_{dimp})$, which makes direct comparison with $\expect(w_{sisa})$ challenging, we present its theoretical upper bound in a form analogous to $\expect(w_{sisa})$, which can also be expressed as $\expect(w_{dimp}) \le p_{uc} \cdot \expect(w_{sisa})$.

This expression shows that the processing speed of DIMP is at least $\frac{1}{p_{uc}}$ times that of SISA.
Our experimental results indicate that the $p_{uc}$ value is typically less than 0.01, meaning that the waiting time for DIMP does not exceed 1\% of SISA. 
We can also qualitatively analyze $p_{uc}$ based on the characteristics of the model:\\
1. The higher the model accuracy, the lower $p_{uc}$. This is because the lower the model accuracy, the more random the prediction results, and each constituent model may give different prediction results; while the higher the model accuracy, the more likely it is that the prediction results of all constituent models are the correct label. According to the definition of fine-grained certification, the higher the consistency among constituent models, the more likely it is to obtain a certified prediction result.\\
2. The longer the single retraining duration $r$, the higher $p_{uc}$. In the case of the same density of unlearnt requests, a larger $r$ usually results in more constituent models being retrained at the same time, which obviously leads to a higher probability of uncertified prediction results.




\section{Experiments}
\label{sec.experiment}
In this section, we conduct experiments to answer the following research questions:
{\bf RQ1}: Can \projectname ~reduce the inference latency of MLaaS and computational overhead while introducing no extra risk?
{\bf RQ2}: How do various groups of parameters affect the performance of \projectname, including 1) the hyper-parameter settings in \projectname: the number of shards and threshold of uncertification ratio; 2) frequency of requests: the density of requests per time slot and the ratio between the two types of requests; 3) server computing capability: the parallel capacity specifying the number of constituent models that can be unlearning-updated in parallel?
{\bf RQ3}: How does the temporal distribution of inference requests and unlearning requests affect the performance of \projectname?

\begin{table}[htbp]
    \footnotesize
    \caption{Details of Datasets and Models}
    \label{table: datasets&models}
    \centering
    \setlength{\tabcolsep}{4pt}
    \begin{tabular}{ccccc}
      \toprule
      \textbf{Dataset} & \textbf{Dimension} & \textbf{Size} & \textbf{Classes} & \textbf{Architecture} \\
      \midrule
      MNIST\cite{DBLP:journals/pieee/LeCunBBH98} & $28 \times 28$ & $60000$ & $10$ &{2 Conv. + 2 FC}\\
      Purchase\cite{DBLP:journals/nca/SakarPKK19} & $600$ & $250000$ & $2$ & 2 FC layers\\
      SVHN\cite{Netzer11Reading} & $32 \times 32 \times 3$ & $604833$ & $10$ & Wide ResNet-1-1\\
      Imagenette\cite{fastai_imagenette} & $500 \times 500 \times 3$ & $13394$ & $10$ & ResNet-18\\
      \bottomrule
    \end{tabular}
\end{table}


\subsection{Experiment Setup}
\partitle{Datasets \& Models}
We evaluate our \projectname ~on four open benchmark datasets and four common model architectures, aligning with mainstream settings, as summarized in Table \ref{table: datasets&models}. 
The datasets exhibit diversities in input dimensions, data volume, and the number of classes, while the models have different sizes and structures. 



\partitle{Baseline Method}
We take SISA with the unlearning-request-first strategy as the baseline. This is modified from the original SISA to avoid extra privacy risk under MLaaS, as described in Sec.\ref{subsec.toy.threats}.

\partitle{Implementation Details}
1) Training: 
in the experiment, we first follow the SISA \cite{bourtoule2021machine} method to divide the training set into multiple shards and then train constituent models separately using the data of each shard. 
The same models are served for all methods to ensure fairness.
2) Inference: 
the models are deployed to respond to requests.
We then randomly generate 500 unlearning requests and a certain number (determined by a hyper-parameter) of inference requests.
These requests are randomly submitted to the server within a time interval, the length of which is dictated by a predefined density parameter.
The same requests are fed to all methods for fairness.
3) Unlearning mechanism details:
once an unlearning update is triggered, the server retrains the constituent models with updated data shards.

\partitle{Evaluation Metrics}
The server processes these requests within the \projectname ~framework, calculating the inference latency and computation overhead.
\begin{itemize}[leftmargin=*]
    \item Average Waiting Time (AWT): The average time interval, measured in \textbf{seconds}, between the submission of an inference request and the receipt of the prediction result, which serves as an indicator of the inference latency.
    To prevent some lines from being crowded together, we applied a \textbf{logarithmic} transformation to the time in the figures, so negative results may be seen when the waiting time length is shorter than 1s.
    \item Number of Retraining (NoR): The total number of times each constituent model has been trained is counted to represent the computation overhead.

\end{itemize}

\begin{table*}[t]
  \caption{Overview of Results on Four Datasets.}
  \scriptsize
  \begin{center}
   \setlength{\tabcolsep}{1.7mm}{
    \begin{tabular*}{\textwidth}{ @{\extracolsep{\fill}} cccccccccc}
     \toprule
     \multirow{1.1}{*}{\textbf{Dataset}} &\multirow{1.1}{*}{\textbf{Metric}}  
     & SISA & DIMP & SUTP & DUTP  & STTU & DTTU & STTP & DTTP \\
     \midrule
     \multirow{2}{*}{PURCHASE}
     &AWT & 14.73 & 0.01($\times$1290.73) &0.57($\times$26.02) & 0.14($\times$105.64) & 0.46($\times$32.25) &0.15($\times$99.41) & 3.31($\times$4.45) & 3.14($\times$4.70)
  \\
     &NoR & 500	& 500($\times$1.000) & 412($\times$0.824) & 412($\times$0.824) & 399($\times$0.798) & 399($\times$0.798) & 350($\times$0.700) &350($\times$0.700)
 \\
     \midrule
     \multirow{2}{*}{SVHN}
    &AWT & 107.55 & 0.17($\times$647.44) &3.40($\times$31.63) & 0.74($\times$145.22) & 2.78($\times$38.70) &0.67($\times$159.82) &16.40($\times$6.56) &16.01($\times$6.72)
  \\
     &NoR & 500 & 500($\times$1.000) &429($\times$0.858) &429($\times$0.858) &421($\times$0.842) &421($\times$0.842) & 345($\times$0.690) &345($\times$0.690)
 \\
      \midrule
     \multirow{2}{*}{MNIST}
    &AWT & 21.34& 0.03($\times$801.63) &0.63($\times$34.09) &0.13($\times$159.86) &0.54($\times$39.76) &0.12($\times$181.95) &3.04($\times$7.03) &2.85($\times$7.49)
  \\
     &NoR & 500	& 500($\times$1.000)& 447($\times$0.894) &447($\times$0.894) &434($\times$0.868) & 434($\times$0.868) & 377($\times$0.754) &377($\times$0.754)
\\
      \midrule
     \multirow{2}{*}{IMAGENETTE}
    &AWT & 1387.59 & 14.76($\times$94.01)&175.64($\times$7.90) &45.24($\times$30.67) &147.91($\times$9.38) &39.95($\times$34.73) &536.35($\times$2.59) &521.77($\times$2.66)
 \\
     &NoR & 500& 500($\times$1.000) &476($\times$0.952) &476($\times$0.952) &474($\times$0.948) &474($\times$0.948) & 405($\times$0.810) &405($\times$0.810)
\\
     \bottomrule
   \end{tabular*}}
   \end{center}
   \label{tab: overview}
\end{table*}

\begin{figure}[!t]
    \centering
    \includegraphics[width=0.47\textwidth]{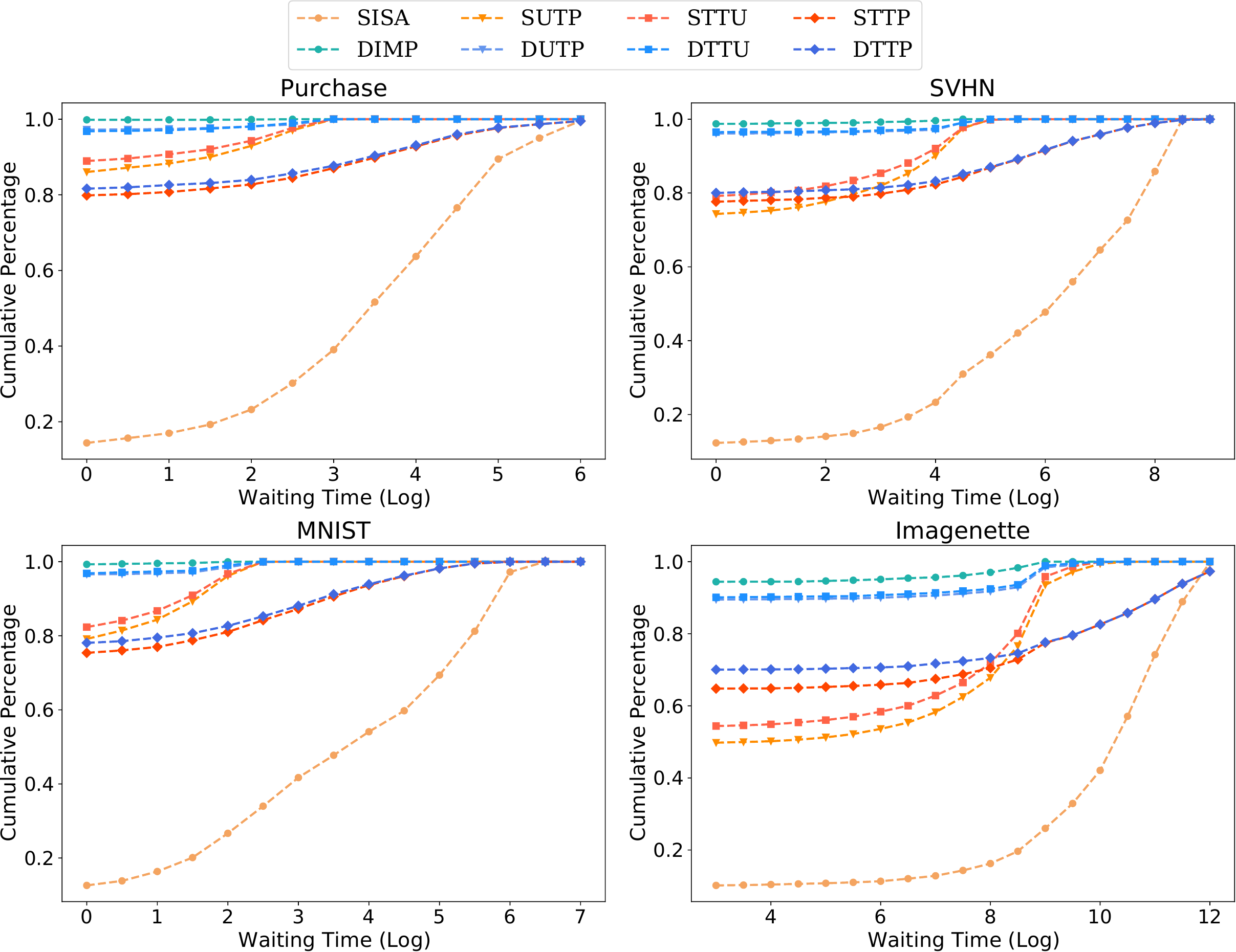}
    \caption{Cumulative Percentage of Waiting Time.}
    \label{fig: cumulative_percentage}
    \Description{This figure presents the cumulative percentage of waiting time for different strategies across four datasets: Purchase, SVHN, MNIST, and Imagenette. Each plot shows the cumulative percentage of waiting time (on a logarithmic scale) on the x-axis, against the percentage of requests processed on the y-axis. The strategies compared include SISA, DIMP, SUTP, STTU, DUTP, DTTU, STTP, and DTTP, each represented by different colored lines and markers. Purchase Dataset: The cumulative waiting time for requests shows a significant variation among the strategies, with SISA (orange dashed line) generally exhibiting longer waiting times compared to other strategies, which converge faster. SVHN Dataset: The patterns are similar, where SISA again has the longest waiting times, while other strategies like DUTP (blue dashed line) and DTTU (blue line) show quicker convergence. MNIST Dataset: The waiting times are more uniform across different strategies, with most lines clustering towards the top, indicating faster processing for most requests. Imagenette Dataset: There is a noticeable spread in waiting times, with SISA showing significantly longer times. Strategies like DUTP and DTTU again show better performance in terms of reduced waiting times. The figure highlights the effectiveness of different strategies in managing and reducing the waiting time for processing requests in machine learning systems, with SISA generally showing poorer performance compared to other methods.}
\end{figure}

\subsection{Overview}\label{sec:rq1}
To answer {\bf RQ1}, we report the overall performance of eight different methods on four real-world datasets in Table \ref{tab: overview}. The numbers in parentheses indicate improvements over SISA (due to space limitations, AWT in the table retains only two decimal places, while the numbers in parentheses are calculated based on AWT before rounding). We also illustrate the distribution of waiting time for all inference requests in the form of cumulative percentage graph in Figure \ref{fig: cumulative_percentage}. In the experiment, we fix the number of shards to 20 and the total number of requests to 5,000, of which 10\% are unlearning requests. The submission time of all requests follows a uniform distribution within a time interval, and the length of the time interval is set to the number of unlearning requests multiplied by the average time required to retrain a single constituent model, which is 6.1s, 22.7s, 4.1s, and 513.4s on 4 datasets with 20 shards, respectively. For methods with a threshold of uncertification ratio, the threshold is set to 0.05.

Regarding inference latency, all methods show significant improvements over SISA, with DIMP achieving up to 1,290 times faster response than SISA. This is primarily due to the majority of inference requests in SISA being blocked by unlearning execution, while most of requests in \projectname ~ can be responded to immediately, as shown in Figure \ref{fig: cumulative_percentage}. In terms of computation overhead, all methods consume less than SISA except DIMP, with STTP and DTTP consuming up to 69\% less computing resources than SISA.
Different methods trade-off on different evaluation metrics, which is consistent with the analysis in Section \ref{subsec.variants}. 

\subsection{Parameters Analysis}
To answer {\bf RQ2}, we use the control variable method to study the influence of various parameters on the performance of \projectname. 
We only show results on Purchase and SVHN due to space limitation. 

\begin{figure}[!t]
    \centering
    \includegraphics[width=0.47\textwidth]{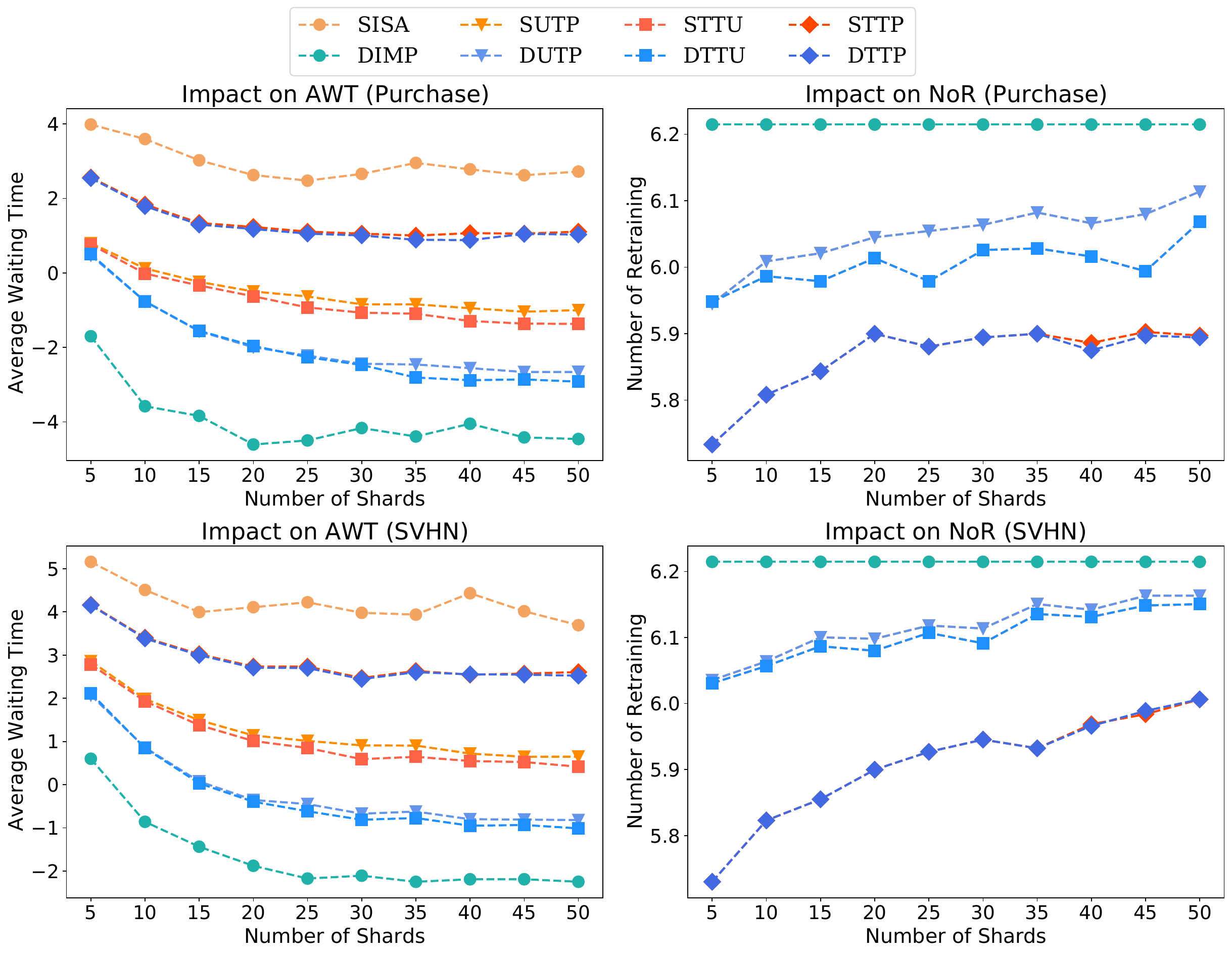}
    \caption{Evaluation on Number of Shards.}
    \label{fig:parameter_shard}
    \Description{This figure evaluates the impact of the number of shards on the average waiting time (AWT) and the number of retraining (NoR) for two datasets, Purchase and SVHN. The figure consists of four plots: Top Left (Impact on AWT for Purchase): This plot shows the average waiting time for different strategies as the number of shards increases from 5 to 50. The strategies include SISA, SUTP, STTU, STTP, DIMP, DUTP, DTTU, and DTTP. SISA generally shows higher waiting times, while strategies like DIMP and DUTP demonstrate significantly lower waiting times as the number of shards increases. Top Right (Impact on NoR for Purchase): This plot displays the number of retraining events required for different strategies as the number of shards increases. SISA consistently shows the highest number of retraining events, while other strategies like DTTP and DTTU show relatively fewer retraining events. Bottom Left (Impact on AWT for SVHN): This plot illustrates the average waiting time for different strategies with the SVHN dataset as the number of shards increases. Similar to the Purchase dataset, SISA has higher waiting times, while DIMP and DUTP show much lower waiting times. Bottom Right (Impact on NoR for SVHN): This plot shows the number of retraining events for different strategies with the SVHN dataset. SISA remains the highest, whereas DTTP and DTTU demonstrate fewer retraining events. Overall, the figure highlights the efficiency of different strategies in terms of reducing average waiting time and the number of retraining events, with DIMP and DUTP performing well compared to SISA.}
\end{figure}

\partitle{Number of Shards}
We demonstrate the impact of the number of shards on AWT and NoR in Figure \ref{fig:parameter_shard}. 

In general, as the number of shards increases, inference latency decreases. 
This is because given the same number of unlearning requests, having more constituent models reduces the proportion of models needing updates, thus increasing the probability of meeting the certified consistency condition.
However, when there are too many shards (e.g. more than 30), the decrease in AWT becomes less significant. This is attributed to the fact that each constituent model would have lower accuracy, resulting in more uncertified prediction results.

The NoR of SISA and DIMP is fixed to the number of unlearning requests, as each unlearning request triggers an unlearning update. 
The NoR of other methods is positively correlated with the number of shards, because when there are numerous shards, the distribution of unlearning requests among constituent models becomes more dispersed, reducing the average number of unlearning requests executed during a single unlearning update.
Note that although there are eight methods in total, only four lines are visible in the figure, as lines of methods that differ only in Option I overlap.

\begin{figure}[!t]
    \centering
    \includegraphics[width=0.47\textwidth]{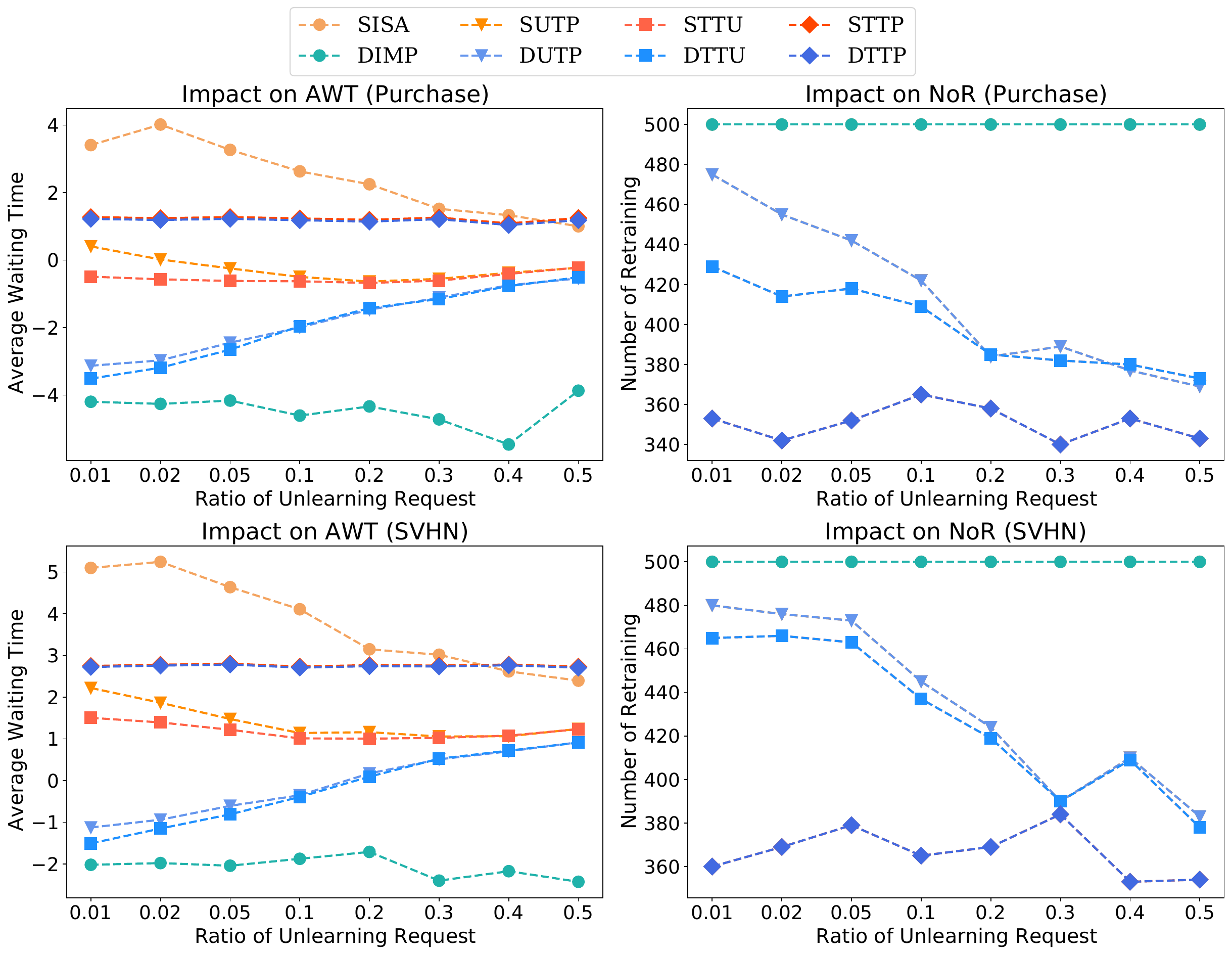}
    \caption{Evaluation on Ratio of Unlearning Requests.}
    \Description{This figure evaluates the impact of the ratio of unlearning requests on the average waiting time (AWT) and the number of retraining (NoR) for two datasets, Purchase and SVHN. The figure consists of four plots: Top Left (Impact on AWT for Purchase): This plot shows the average waiting time for different strategies as the ratio of unlearning requests increases from 0.01 to 0.5. The strategies include SISA, SUTP, STTU, STTP, DIMP, DUTP, DTTU, and DTTP. SISA generally shows higher waiting times, while strategies like DIMP and DUTP demonstrate significantly lower waiting times as the ratio of unlearning requests increases. Top Right (Impact on NoR for Purchase): This plot displays the number of retraining events required for different strategies as the ratio of unlearning requests increases. SISA consistently shows the highest number of retraining events, while other strategies like DTTP and DTTU show relatively fewer retraining events. Bottom Left (Impact on AWT for SVHN): This plot illustrates the average waiting time for different strategies with the SVHN dataset as the ratio of unlearning requests increases. Similar to the Purchase dataset, SISA has higher waiting times, while DIMP and DUTP show much lower waiting times. Bottom Right (Impact on NoR for SVHN): This plot shows the number of retraining events for different strategies with the SVHN dataset. SISA remains the highest, whereas DTTP and DTTU demonstrate fewer retraining events. Overall, the figure highlights the efficiency of different strategies in terms of reducing average waiting time and the number of retraining events, with DIMP and DUTP performing well compared to SISA.}
  \label{fig:unlearning_evaluation}
    \label{fig:parameter_ratio}
\end{figure}

\partitle{Ratio of Unlearning Requests}
To investigate the impact of the ratio of unlearning requests on \projectname's performance, we adjust the ratio between two requests with the number of unlearning requests fixed at 500 and the number of inference requests determined by the ratio.
The results are shown in Figure \ref{fig:parameter_ratio}.

The NoR of SUTP and DUTP decreases as the ratio increases because a higher number of inference requests increases the likelihood of yielding at least one uncertified prediction result, leading to more frequent unlearning updates.
However, the computation overhead of STTP and DTTP is almost unaffected by the ratio, as their unlearning update triggering condition depends on the uncertification proportion of all the inference requests, rather than the absolute number of uncertified inference requests.

The inference latency of SISA decreases with the decrease of inference requests, as too many inference requests prevent the server from fully utilizing its parallel capacity. 
For instance, if two unlearning requests arrive consecutively, the server can retrain their corresponding constituent models in parallel; however, if an inference request is inserted between two unlearning requests, the inference request must wait for the first unlearning request to be executed, and the second unlearning request has to wait for the processing of the inference request, making the execution of these two unlearning requests sequential. 
For methods such as SUTP and STTU that maintain a single context, more frequent unlearning updates lead to longer AWT. 
For DUTP and DTTU, which maintain both contexts, although the number of uncertified inferences increases with the total number of inference requests, the uncertification proportion decreases. 
In other words, the proportion of inference requests that need to wait decreases, resulting in a reduction in AWT. The inference latency of the remaining three methods is almost unaffected by the ratio of unlearning requests.

\begin{figure}[!t]
    \centering
    \includegraphics[width=0.47\textwidth]{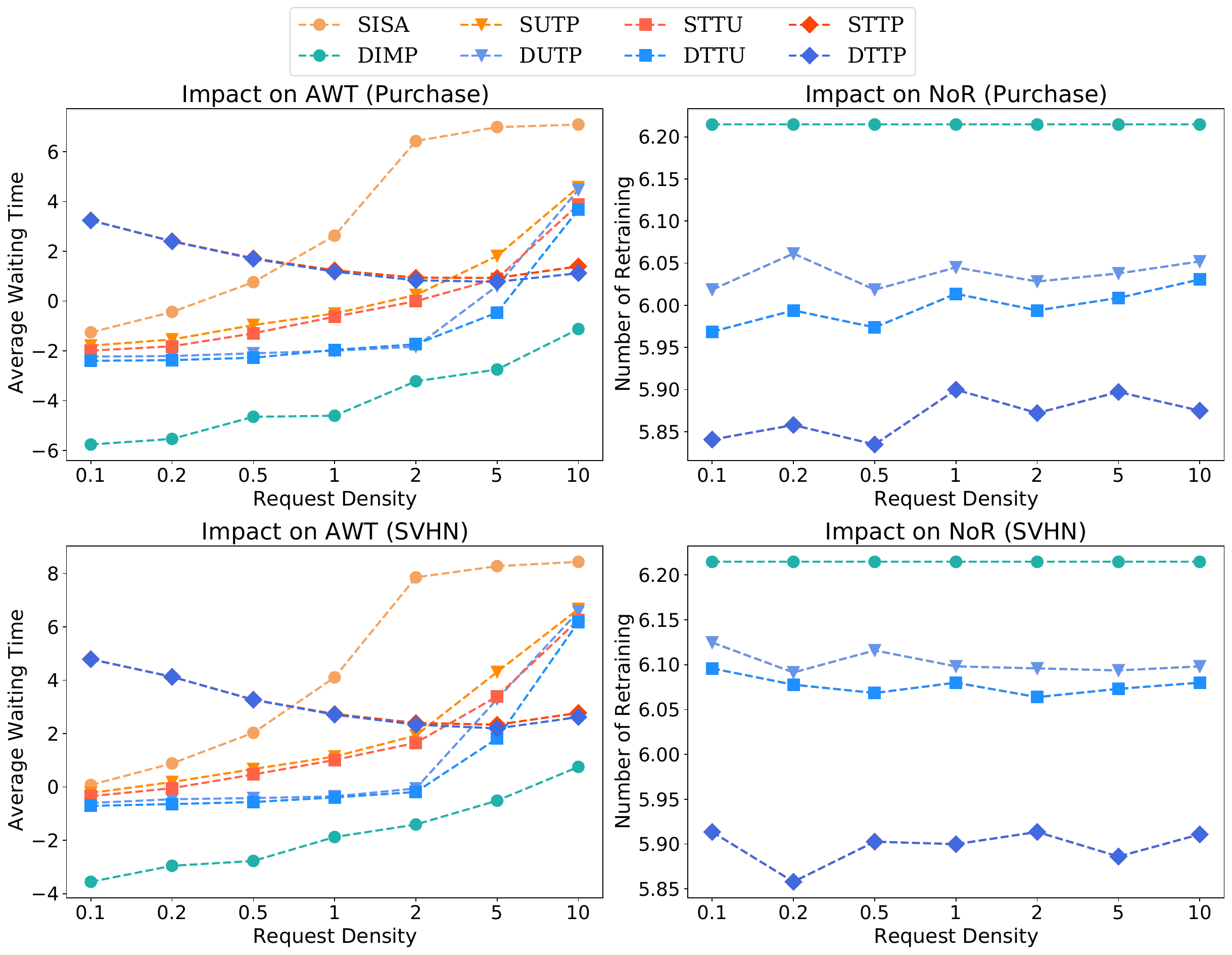}
    \caption{Evaluation on Request Density.}
    \label{fig:parameter_density}
      \Description{This figure evaluates the impact of request density on the average waiting time (AWT) and the number of retraining (NoR) for two datasets, Purchase and SVHN. The figure consists of four plots: Top Left (Impact on AWT for Purchase): This plot shows the average waiting time for different strategies as the request density increases from 0.1 to 10. The strategies include SISA, SUTP, STTU, STTP, DIMP, DUTP, DTTU, and DTTP. SISA generally shows higher waiting times, while strategies like DIMP and DUTP demonstrate significantly lower waiting times as the request density increases. Top Right (Impact on NoR for Purchase): This plot displays the number of retraining events required for different strategies as the request density increases. SISA consistently shows the highest number of retraining events, while other strategies like DTTP and DTTU show relatively fewer retraining events. Bottom Left (Impact on AWT for SVHN): This plot illustrates the average waiting time for different strategies with the SVHN dataset as the request density increases. Similar to the Purchase dataset, SISA has higher waiting times, while DIMP and DUTP show much lower waiting times. Bottom Right (Impact on NoR for SVHN): This plot shows the number of retraining events for different strategies with the SVHN dataset. SISA remains the highest, whereas DTTP and DTTU demonstrate fewer retraining events. Overall, the figure highlights the efficiency of different strategies in terms of reducing average waiting time and the number of retraining events, with DIMP and DUTP performing well compared to SISA.}
\end{figure}

\begin{table*}[t]
  \caption{\small Impact of Request Distribution on Imagenette.}
  \scriptsize
  \begin{center}
   \setlength{\tabcolsep}{1.7mm}{
    \begin{tabular*}{\textwidth}{ @{\extracolsep{\fill}} cccccccccc}
     \toprule
     \multirow{1.1}{*}{\textbf{Metric}} &\multirow{1.1}{*}{\textbf{Distribution}}  
      & SISA & DIMP & SUTP & DUTP & STTU & DTTU & STTP & DTTP   \\
     \midrule
     \multirow{12}{*}{AWT}
     &unimodal & 7526.3 & 16.1($\times$467.6) & 207.1($\times$36.3)& 49.5($\times$152.0) & 175.7($\times$42.8) & 43.4($\times$173.3) & 557.8($\times$13.5) & 542.2($\times$13.9)  \\
     &unimodal($\mu=\frac{T}{4}$) & 21373.9 & 16.6($\times$1291.1) & 241.5($\times$88.5) & 54.7($\times$391.0) & 212.4($\times$100.7) & 50.0($\times$427.5) &632.4($\times$33.8) & 609.4($\times$35.1) \\
     &unimodal($\mu=\frac{3T}{4}$) & 16157.9 & 16.88($\times$957.3) & 256.6($\times$63.0) & 54.7($\times$295.2) & 226.9($\times$71.2) & 50.4($\times$320.9) & 443.4($\times$36.4) & 431.0($\times$37.5) \\
     &unimodal($\sigma=\frac{T}{2}$) & 1715.2 & 14.73($\times$116.4) & 166.5($\times$10.3) & 44.5($\times$38.6) & 157.5($\times$10.9) & 42.5($\times$40.4) & 531.5($\times$3.2) & 514.7($\times$3.3)  \\
     &unimodal($\sigma=\frac{T}{4}$) & 22274.0 & 17.2($\times$1298.9) & 252.4($\times$88.3) & 56.6($\times$393.6) & 2312.0($\times$96.0) & 54.6($\times$408.0) & 506.4($\times$44.0) & 489.4($\times$45.5) \\
     \cmidrule(lr){2-10}
     &bimodal & 5522.1 & 14.2($\times$388.4) & 190.0($\times$29.01 &46.8($\times$118.0) &166.7($\times$33.12) &42.8($\times$129.1) & 508.8($\times$10.9) & 492.2($\times$11.2)  \\
     &trimodal & 3649.7 & 15.8($\times$231.5) & 185.6($\times$19.7) & 47.0($\times$77.7) &175.3($\times$20.8)&44.4($\times$82.1)& 548.4($\times$6.7)& 530.7($\times$6.9)\\
     &quadmodal & 1907.7 & 15.7($\times$121.3) & 180.0($\times$10.60) & 45.8($\times$41.6) & 168.7($\times$11.3) & 43.4($\times$44.1)& 559.4($\times$3.4) & 543.8($\times$3.5)\\
     \cmidrule(lr){2-10}
     &diff-$\mu$($\mu_u=\frac{T}{4}$, $\mu_i=\frac{3T}{4}$) &
     6754.8 & 10.9($\times$618.1) &144.6($\times$46.7)& 36.87($\times$183.2) & 134.3($\times$50.3) & 35.9($\times$188.2) & 406.2($\times$16.6)& 387.6($\times$17.4)\\
     &diff-$\mu$($\mu_u=\frac{3T}{4}$, $\mu_i=\frac{T}{4}$) &
     3144.5	&12.0($\times$262.4) &133.9($\times$23.5)& 36.8($\times$85.5) & 126.1($\times$24.9) & 34.9($\times$90.1) & 672.2($\times$4.7) &  660.0($\times$4.8)\\
     &diff-$\sigma$($\sigma_u=\frac{T}{2}$, $\sigma_i=\frac{T}{4}$) & 2079.9 &13.3($\times$156.7) &164.0($\times$12.7)&41.9($\times$49.6)&147.7($\times$14.1)&39.3($\times$52.9)&553.9($\times$3.8)&538.2($\times$3.9)  \\
     &diff-$\sigma$($\sigma_u=\frac{T}{4}$, $\sigma_i=\frac{T}{2}$) & 14436.8 &15.0($\times$961.9) &213.6($\times$67.6)&55.5($\times$260.3)&188.5($\times$76.6)&50.7($\times$284.7)&517.2($\times$27.9)&506.2($\times$28.5) \\
     \midrule
     \multirow{12}{*}{NoR}
    &unimodal & 500	&500($\times$1.000) &473($\times$0.946)&473	($\times$0.946) &470($\times$0.940)&470($\times$0.940) &401($\times$0.802) &401($\times$0.802)  \\
     &unimodal($\mu=\frac{T}{4}$) & 500	&500($\times$1.000) &480($\times$0.960) & 480($\times$0.960) &477($\times$0.954) & 477($\times$0.954)& 406($\times$0.812) & 406($\times$0.812)\\
     &unimodal($\mu=\frac{3T}{4}$) & 500	& 500($\times$1.000) & 475($\times$0.950) & 475($\times$0.950) & 473($\times$0.946)& 473($\times$0.946) & 401($\times$0.802) & 401($\times$0.802)
  \\
     &unimodal($\sigma=\frac{T}{2}$) & 500 &500($\times$1.000)& 478($\times$0.956) &478($\times$0.956) &478($\times$0.956) &478($\times$0.956) &416($\times$0.832) &416($\times$0.832)
 \\
     &unimodal($\sigma=\frac{T}{4}$) & 500 & 500($\times$1.000)& 465($\times$0.930 & 465($\times$0.930) & 463($\times$0.926)& 463($\times$0.926)& 393($\times$0.786) & 393($\times$0.786)
\\
     \cmidrule(lr){2-10}
     &bimodal & 500 & 500($\times$1.000) & 475($\times$0.950)& 475($\times$0.950)& 472($\times$0.944) & 472($\times$0.944) & 402($\times$0.804) &402($\times$0.804)
   \\
     &trimodal & 500 & 500($\times$1.000) & 479($\times$0.958) & 479($\times$0.958) & 476($\times$0.952) & 476($\times$0.952) & 412($\times$0.824) & 412($\times$0.824)
  \\
     &quadmodal & 500 & 500($\times$1.000) & 482($\times$0.964) & 482($\times$0.964) & 480($\times$0.960) & 480($\times$0.960) & 413($\times$0.826) & 413($\times$0.826)
 \\
     \cmidrule(lr){2-10}
     &diff-$\mu$($\mu_u=\frac{T}{4}$, $\mu_i=\frac{3T}{4}$) & 500	& 500($\times$1.000) &457($\times$0.914) & 457($\times$0.914) &454($\times$0.908) & 454($\times$0.908) &414($\times$0.828) &414($\times$0.828)
 \\
     &diff-$\mu$($\mu_u=\frac{3T}{4}$, $\mu_i=\frac{T}{4}$) & 500	& 500($\times$1.000) & 464($\times$0.928) & 464($\times$0.928) & 461($\times$0.922) & 461($\times$0.922) & 319($\times$0.638) & 319($\times$0.638)
  \\
     &diff-$\sigma$($\sigma_u=\frac{T}{2}$, $\sigma_i=\frac{T}{4}$) & 500	& 500($\times$1.000) &478($\times$0.956) & 478($\times$0.956) &475($\times$0.950) &475($\times$0.950) & 412($\times$0.824) &412($\times$0.824)
  \\
     &diff-$\sigma$($\sigma_u=\frac{T}{4}$, $\sigma_i=\frac{T}{2}$) & 500	& 500($\times$1.000) & 474($\times$0.948) &474($\times$0.948) &467($\times$0.934) &467($\times$0.934) & 401($\times$0.802) &401($\times$0.802)
  \\
     \bottomrule
   \end{tabular*}}
   \end{center}
   \label{tab: distribution}
\end{table*}

\partitle{Request Density}
In the experiments in Sect.\ref{sec:rq1}, we set the basic length of the time interval to the number of unlearning requests multiplied by the time required to retrain a single constituent model. 
To explore the impact of request density on \projectname's performance, we divide the basic time interval length by a coefficient; the larger the coefficient, the shorter the time interval, and the denser the arrival of requests. The effect of Request Density ~is shown in Figure \ref{fig:parameter_density}. The NoR of all methods is not significantly related to density.

The AWT of all methods except STTP and DTTP escalates as the density increases because the overall time interval becomes shorter. Although the total duration of unlearning execution remains nearly constant, its proportionate share increases, leading to a higher likelihood that inference requests will have to wait.
However, STTP and DTTP show the opposite trend, which is due to the composition of their waiting time. 
The waiting time for inference requests in these two methods primarily depends on waiting for the uncertification ratio to reach the threshold.
Therefore, the higher the request density, the faster the uncertification ratio reaches the threshold, and thus the waiting time decreases. 


\begin{figure}[!t]
    \centering
    \includegraphics[width=0.47\textwidth]{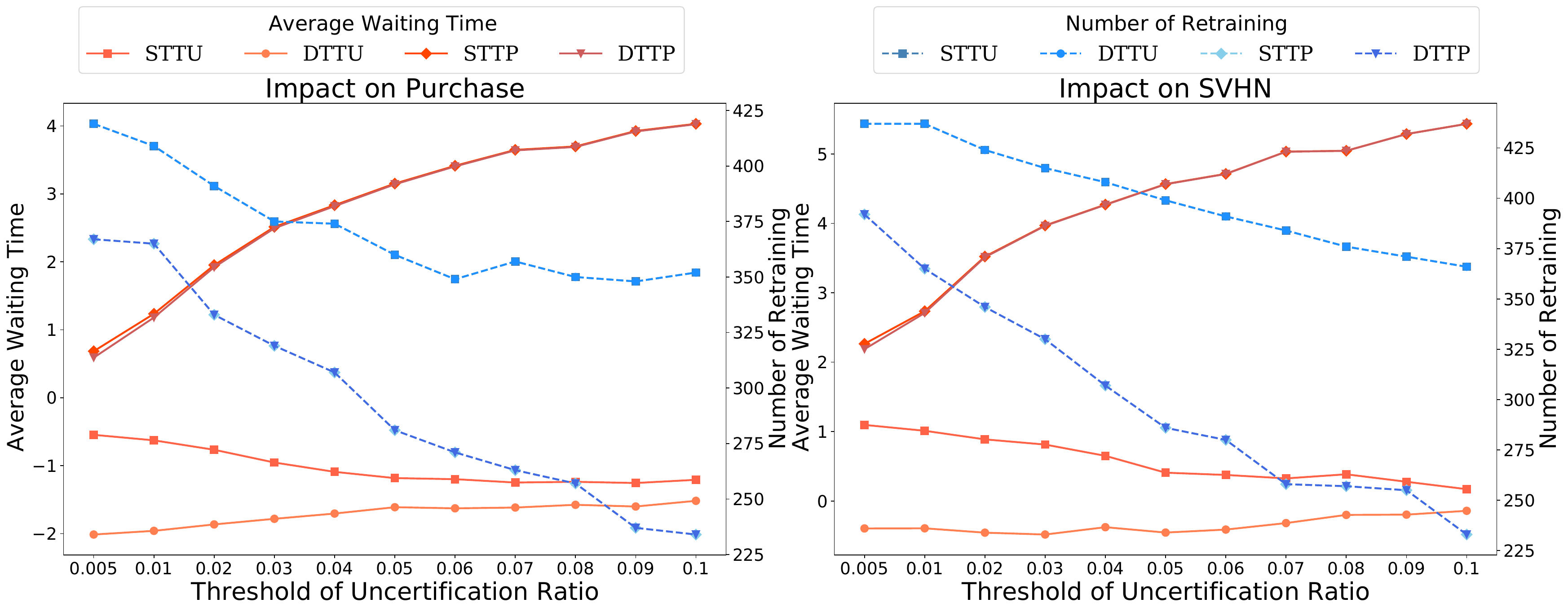}
    \caption{Evaluation on Threshold of Uncertification Ratio.}
    \label{fig:parameter_uncertified}
    \Description{This figure evaluates the impact of the threshold of uncertification ratio on the average waiting time (AWT) and the number of retraining (NoR) for two datasets, Purchase and SVHN. The figure consists of two plots: Left (Impact on Purchase): This plot shows the average waiting time and the number of retraining events for different strategies as the threshold of uncertification ratio increases from 0.005 to 0.1. The strategies include STTU, DTTU, STTP, and DTTP. The average waiting time is indicated by lines with circular markers, while the number of retraining events is indicated by lines with square markers. STTU generally shows a higher number of retraining events and a decreasing trend in waiting time, while DTTP shows the lowest number of retraining events with a relatively stable waiting time. Right (Impact on SVHN): This plot illustrates the same metrics for the SVHN dataset. Similar patterns are observed, with STTU showing a higher number of retraining events and decreasing waiting time, whereas DTTP maintains a low number of retraining events and stable waiting times. Overall, the figure highlights how varying the threshold of uncertification ratio affects the performance of different strategies in terms of average waiting time and the number of retraining events, with DTTP performing consistently well compared to others.}
\end{figure}

\partitle{Threshold of Uncertification Ratio}
The threshold-triggered methods initiate an unlearning update once the uncertification threshold is exceeded.
We show the impact of the threshold on the performance of these four methods in Figure \ref{fig:parameter_uncertified}. 
To save space, we plot AWT and NoR in the same figure.

The NoR and threshold of these four methods are negatively correlated because a higher threshold requires more uncertified predictions to trigger retraining.
The inference latency of STTP and DTTP increases with an increasing threshold, as inference requests in the uncertification list must first wait for the ratio to reach the threshold, and a larger threshold results in a longer waiting time.
The AWT of STTU decreases with an increasing threshold, because a higher threshold reduces NoR, and the number of inference requests that need to wait for the completion of unlearning updates also decreases.
The AWT of DTTU remains almost unchanged with the increasing threshold, possibly because the server can still respond to inference requests during unlearning updates. 

\subsection{Temporal Distribution of Requests}
In previous experiments, we assume that both types of requests are uniformly distributed over time. However, this may not be the case in real-world scenarios as request distribution may be uneven, and different types of requests might follow distinct distributions.
We use the Gaussian distribution as an example to study the effects of the temporal distribution of requests on the performance of \projectname ~to answer {\bf RQ3}. 
We focus on three distribution characteristics:
1. The mean and variance; 
2. Unimodal or multimodal;
3. Differences between unlearning and inference request distributions.
Assuming the time interval is $[0, T]$, if the sampled arrival time lies outside the interval, we will re-sample. 
Experimental results on the Imagenette dataset are shown in Table \ref{tab: distribution}.

Initially, we assume that both unlearning and inference requests follow a Gaussian distribution with a mean of $\mu=\frac{T}{2}$ and a standard deviation of $\sigma=\frac{T}{3}$. The first row of the table displays the results. The inference latency increases compared to the uniform distribution in Table \ref{tab: overview}, while the NoR slightly decreases. We then alter the mean to $\mu=\frac{T}{4}$ and $\mu=\frac{3T}{4}$, corresponding to rows 2-3. The inference latency increases further, potentially due to the resampling strategy. When the distribution's center is close to the interval's edge, resampling points outside the interval cause the peak to be higher. Next, we change the standard deviation to $\sigma=\frac{T}{2}$ and $\sigma=\frac{T}{4}$, corresponding to rows 4-5. A larger standard deviation flattens the distribution.

Considering that the temporal distribution of requests may have multiple peaks, we assume several multimodal Gaussian distributions, corresponding to rows 6-8. 
The results show that as the number of peaks increases, the inference latency decreases and the computation overhead increases, because a distribution with more peaks resembles a uniform distribution more closely.

Lastly, we examine cases where unlearning and inference requests follow different distributions. We first consider different mean values (rows 9-10). Regardless of whether the unlearning or inference request distribution peak came first, both inference latency and computation overhead decrease, sometimes even surpassing uniform distribution performance. This is because inference requests are more likely to wait while numerous unlearning requests are submitted concurrently. Staggering the peaks of the two distributions results in fewer waiting inference requests and reduces unlearning update frequency, enabling single unlearning updates to process more requests. However, when the inference request distribution peak comes first, i.e., $\mu_u=\frac{3T}{4}$ and $\mu_i=\frac{T}{4}$, the inference latency for STTP and DTTP becomes significantly high. This is because some inference requests' prediction results may be uncertified at the peak, but the ratio might not reach the threshold, forcing them to wait until the unlearning request peak for reprocessing. We also consider different standard deviations (rows 11-12). A Higher concentration of unlearning requests leads to increased inference latency, while the impact of a higher concentration of inference requests is relatively small.

Please refer to the appendices in arXiv version for additional experimental results, including the scenarios where requesters might have malicious intentions. 

\section{Conclusion and Future Work}
\label{sec.conclusion}
In this paper, we reveal that handling unlearning requests in an inference-oblivious manner under MLaaS setting can introduce two new security and privacy vulnerabilities. We proposed a pioneering framework, \projectname, for inference serving-aware machine unlearning that simultaneously serves both inference requests from queriers and unlearning requests from data owners. \projectname ~devised a novel certified inference consistency approach to reduce inference latency without introducing extra privacy risks for unlearning requested data. It also offers three groups of design options, providing seven specific variants tailored to different MLaaS environments and design objectives. Our extensive empirical results demonstrate the effectiveness of \projectname ~across various settings.


Future work includes 1) Studying more optimal design choices (e.g., better unlearning timing) that can adapt to the workloads of inference and unlearning requests; 2) Considering other critical factors (e.g., resources provision and cost) in terms of their security and privacy implications, as well as how to achieve better Pareto frontiers among all factors; 3) For certain specific ML problems, employing corresponding machine unlearning algorithms instead of retraining to update the constituent models.

\begin{acks}
This work is supported by the National Key Research and Development Program of China (No. 2022YFE0113200), the National Natural Science Foundation of China (No. 62206207 \& No. 62072395), and the Joint Fund of the National Natural Science Foundation of China (No. U20A20178).
\end{acks}



\bibliographystyle{ACM-Reference-Format}
\balance
\bibliography{reference}
\newpage

\appendix
\section{Deferred proof of Theorem \ref{thm.certified.inference.consistency}}
\label{appendix.proof.eraser,consistency}
\begin{proof} 
    According to the pending unlearning requests, We have the following two cases: \\
    \emph{Case I.} $\tS_k^t = \tS_k^{\tt O}$: There is no pending unlearning requests to Shard $k$, which implies $f_k^{t}(\z) = f_k^{\tt O}(\z)$ upon step $t$; \\
    \emph{Case II.} $\tS_k^t \neq \tS_k^{\tt O}$: there is one or multiple pending unlearning requests to Shard $k$ upon timestamp $t$, which has the potential of $f_k^{t}(\z) \neq f_k^{\tt O}(\z)$. 
    
    According to the prediction results, Case II can be further divided into three subcases. The aim is to estimate whether the largest possible ${\tt Count}_{y_b}^t(\z)$ will be greater than the smallest possible ${\tt Count}_{y_a}^t(\z)$, which is the most likely situation (i.e., the easiest case for inference inconsistency) to flip the final prediction result from $y_a$ to $y_b$ and the inference consistency no longer holds. For any $\{k\in[K] \big{|} \tS_k^t \neq \tS_k^{\tt O}\}$, we have three possible prediction results for $f_k^{\tt O}(\z)$ and the potential $f_k^t(\z)$:
    \begin{enumerate}[leftmargin=*]
        \item if $f_k^{\tt O}(\z) = y_a$, and the potential $f_k^t(\z) = y_b$;
        \item if $f_k^{\tt O}(\z) = y_b$, and the potential $f_k^t(\z) = y_b$;
        \item $f_k^{\tt O}(\z) \neq y_a$ \& $f_k^{\tt O}(\z) \neq y_b$, and the potential $f_k^t(\z) = y_b$.
    \end{enumerate}
    The counts of the above three subcases correspond exactly to $\gamma_1, \gamma_{2}, \gamma_{3}$, which can be collected based solely on the prediction results from $f_{1}^{\tt O}(\z),\dots,f_{K}^{\tt O}(\z)$. Then, we have the lower bound of ${\tt Count}_{y_a}^t(\z)$ to be
    \begin{equation}
        {\tt Count}_{y_a}^t(\z) \leq {\tt Count}_{y_a}^{\tt O}(\z) - \gamma_1,
    \end{equation}
    and the upper bound of ${\tt Count}_{y_b}^t(\z)$ to be
        \begin{equation}
        {\tt Count}_{y_b}^t(\z) \geq {\tt Count}_{y_b}^{\tt O}(\z) + \gamma_1 +\gamma _3.
    \end{equation}
    As a result, as long as for all $y_b \neq y_a$, ${\tt Count}_{y_a}^{\tt O}(\z) - \gamma_1 > {\tt Count}_{y_b}^{\tt O}(\z) + \gamma_1 +\gamma _3$, which can be relaxed to $2\gamma _1 + \gamma _3 \leq  {\tt Count}_{y_a}^{\tt O}(\z) - {\tt Count}_{y_b}^{\tt O}(\z) - \mathbb{I}(y<y_a) = \Gamma$, we have $F^t(\z) = F^{\tt O}(\z)$.
\end{proof}

\section{Deferred proof of Theorem \ref{theorem:sisa}}
\label{appendix:proof.sisa}

\begin{proof} 
    For a certain inference request, let its arrival time at the server be $t_i$. Assume it is between the arrival time of the $i$-th unlearning request and the $(i+1)$-th unlearning request, i.e., $t_i \in [(i-1)\cdot\frac{T}{n_u}, i\cdot\frac{T}{n_u})$. The retraining time for the $i$-th unlearning request is $[(i-1)\cdot\frac{T}{n_u},(i-1)\cdot\frac{T}{n_u}+r]$.
    
    Based on whether the time required to train a constituent model, $r$, is greater than the arrival interval of unlearning requests, $\frac{T}{n_u}$, we have two cases:
    
    \textbf{1. $r \le \frac{T}{n_u}$.}In this case, the waiting time for the prediction, $w_{sisa}$, can be expressed as:
    \begin{equation}
        w_{sisa}=
        \begin{cases}
            \begin{aligned}
                &(i-1)\cdot\frac{T}{n_u}+r-t_i, \\
                &\quad\quad\quad\quad t_i\in [(i-1)\cdot\frac{T}{n_u},(i-1)\cdot\frac{T}{n_u}+r]
            \end{aligned}\\ 
            \begin{aligned}
                &0, \\
                &\quad\quad\quad\quad t_i\in[(i-1)\cdot\frac{T}{n_u}+r, i\cdot\frac{T}{n_u}].
            \end{aligned}
        \end{cases}
    \end{equation}
    Since $t_i$ follows a uniform distribution, i.e.,
    \begin{equation}
        P(t_i)=\frac{n_u}{T}, \quad t_i\in [(i-1)\cdot\frac{T}{n_u}, i\cdot\frac{T}{n_u})
    \end{equation}
    Therefore, the expected value of $w_{sisa}$ is:
    \begin{align}
        \expect(w_{sisa})&=
        \int_{(i-1)\cdot\frac{T}{n_u}}^{(i-1)\cdot\frac{T}{n_u}+r}
        \frac{n_u}{T}((i-1)\cdot\frac{T}{n_u}+r-t_i) dt_i\\
        &=\frac{n_ur^2}{2T}.
    \end{align}
    
    \textbf{2. $r > \frac{T}{n_u}$.}In this case, the waiting time for the prediction, $w_{sisa}$, can be expressed as:
    \begin{equation}
        w_{sisa}=(i-1)\cdot\frac{T}{n_u}+r-t_i.
    \end{equation}
    The expected value of $w_{sisa}$ is
    \begin{align}
        \expect(w_{sisa})&=
        \int_{(i-1)\cdot\frac{T}{n_u}}^{i\cdot\frac{T}{n_u}}
        \frac{n_u}{T}((i-1)\cdot\frac{T}{n_u}+r-t_i) dt_i\\
        &=r-\frac{T}{2n_u}.
    \end{align}
    To conclude, we have proved that 
    \begin{equation}
        \expect(w_{sisa})=
        \begin{cases}
            \frac{n_ur^2}{2T}, &\quad r\le \frac{T}{n_u}\\
            r-\frac{T}{2n_u}, &\quad r>\frac{T}{n_u}.
        \end{cases}
    \end{equation}
\end{proof}

\section{Deferred proof of Theorem \ref{theorem:DIMP}}
\label{appendix.proof.time.eraser}

\begin{proof}
    The number of constituent models being retrained, $k_r$, when an inference request arrives can be represented as:
    \begin{equation}
        k_r=
        \begin{cases}
            \lfloor \frac{rn_u}{T}\rfloor, &\quad (t_i \mod \frac{T}{n_u})> r \mod (\frac{T}{n_u})\\
            \lfloor \frac{rn_u}{T}\rfloor+1, &\quad (t_i \mod \frac{T}{n_u})\le (r \mod \frac{T}{n_u}).
        \end{cases}
    \end{equation}
    The time difference, $t_d$, between the arrival of the inference request and the nearest completion of retraining can be represented as:
    \begin{equation}
    t_d =
    \begin{cases}
        \begin{aligned}
            &\frac{T}{n_u} + (r \mod \frac{T}{n_u}) - (t_i \mod \frac{T}{n_u}), \\
            &\quad\quad\quad\quad\quad\quad (t_i \mod \frac{T}{n_u}) > (r \mod \frac{T}{n_u})
        \end{aligned}\\
        \begin{aligned}
            &(r \mod \frac{T}{n_u}) - (t_i \mod \frac{T}{n_u}), \\
            &\quad\quad\quad\quad\quad\quad (t_i \mod \frac{T}{n_u}) \le (r \mod \frac{T}{n_u}).
        \end{aligned}
    \end{cases}
\end{equation}
    For each inference request, a maximum of $k_r$ certification judgments will occur, which are when the inference request arrives and when each constituent model finishes retraining (excluding the last finished constituent model, because when all retraining constituent models are completed, the prediction result must be certified).
    
    Therefore, the expected waiting time, $w_{dimp}$, for the prediction can be represented as:
    \begin{align}
        \expect(w_{dimp})=
        &\sum_{i=1}^{k-1}((1-p_{uc})p_{uc}^i\cdot(t_d+\frac{(i-1)T}{n_u}))\\
        &+p_{uc}^k(t_d+\frac{(k-1)T}{n_u}) + (1-p_{uc})\cdot 0\\
        =&p_{uc}\cdot t_d+\sum_{i=2}^{k-1}((1-p_{uc})p_{uc}^i\cdot \frac{(i-1)T}{n_u})\\
        &+p_{uc}^k\frac{(k-1)T}{n_u}.
    \end{align}
    This expression can hardly be simplified into a more intuitive form, but we can find a relatively concise upper bound for the expression: if the prediction result is uncertified when the inference request arrives, we skip all subsequent certification judgments and wait for all the retraining constituent models to complete their retraining before returning the prediction result. Obviously, this would result in a strictly longer waiting time, which is the same as the waiting time of SISA. Therefore, we can derive that:
    \begin{align}
        \expect(w_{dimp})&\le
        (1-p_{uc})\cdot 0+p_{uc}\cdot \expect(w_{sisa})\\&=
        \begin{cases}
            p_{uc}\cdot \frac{n_ur^2}{2T}, &\quad r\le \frac{T}{n_u}\\
            p_{uc}\cdot (r-\frac{T}{2n_u}), &\quad r>\frac{T}{n_u}.
        \end{cases}
    \end{align}
\end{proof}

\section{Straightforward adaptation from \texorpdfstring{\cite{DBLP:conf/iclr/0001F21}}{[DBLP:conf/iclr/0001F21]}}
\label{appendix.straitforward}

\begin{theorem} \label{thm.straitforward.certification} \textbf{(Straightforward adaptation from \cite{DBLP:conf/iclr/0001F21})}
Under the same setting with Theorem \ref{thm.certified.inference.consistency}, let the number of impacted shards be $\gamma$ as defined below,
\begin{equation}
    \gamma := |\{k\in[K] \Big{|} \tS_k^t \neq \tS_k^{\tt O}\}|.
\end{equation}
The model can respond to the inference request of $\z$ without retraining, while ensuring the prediction is exactly the same as the otherwise retrained prediction, i.e., $F^t(\z) = F^{\tt O}(\z) = y_a$, as long as the following condition holds: $\forall y\neq y_a$, 
\begin{equation}
    2\gamma \leq \Gamma,
\end{equation}
where $\Gamma  = \max_{\forall y\neq y_a}\{ {\tt Count}_{y_a}^{\tt O}(\z) - {\tt Count}_{y}^{\tt O}(\z) - \mathbb{I}(y<y_a) \}$.
\end{theorem}
\begin{proof} \textbf{(Straightforward adaptation from \cite{DBLP:conf/iclr/0001F21})}
Our aim is to prove that, as long as $\gamma \leq \Gamma$, we certify that $F^t(\z) = F^{\tt O}(\z)$, i.e., the original model will provide exactly the same prediction label with the otherwise retrained model. Consequently, it allows for serving the inference request on $\z$ without executing the actual data forgotten requests of $\UR^t$ by retraining the model. 
There are two cases: 1) $\tS_k^t = \tS_k^{\tt O}$, there is no data deletion in Shard $k$, which implies $f_k^{t}(\z) = f_k^{\tt O}(\z)$ upon step $t$; 2) $\tS_k^t \neq \tS_k^{\tt O}$, there is one or multiple data deletions in Shard $k$ upon step $t$, which can lead to $f_k^{t}(\z) \neq f_k^{\tt O}(\z)$. Thus, by Case 2, there are at most $\tau$ constituent classifiers that can have $f_k^{t}(\z) \neq f_k^{\tt O}(\z)$. For any $y\in\mathbb{Y}$, we have
\begin{equation}
    |{\tt Count}_{y}^{\tt O}(\z) - {\tt Count}_{y}^t(\z) | \leq \gamma.
\end{equation}
Let $y_a = F^{\tt O}(\z)$. As long as the following condition holds
\begin{align}
    {\tt Count}_{y_a}^t(\z) > {\tt Count}_{y}^t(\z), & \text{if } y < y_a  \\
    {\tt Count}_{y_a}^t(\z) \geq {\tt Count}_{y}^t(\z), & \text{if } y > y_a, 
\end{align}
which can be equivalently described by 
\begin{equation}
\label{eq.unlearned.count.condition}
    \forall y\neq y_a,\ {\tt Count}_{y_a}^t(\z) \geq {\tt Count}_{y}^t(\z) + \mathbb{I}(y<y_a),
\end{equation}
we will have  $F^t(\z) = y_a$. Thus, we have $F^t(\z) = y_a$ if
\begin{equation}
    {\tt Count}_{y_a}^t(\z) \geq {\tt Count}_{y_a}^{\tt O}(\z) - \gamma 
\end{equation}
\begin{equation}
    {\tt Count}_{y}^{\tt O}(\z) + \gamma \geq {\tt Count}_{y}^t(\z)
\end{equation}
In order to have eq.(\ref{eq.unlearned.count.condition}), it suffices to ensure
\begin{equation}
    {\tt Count}_{y_a}^{\tt O}(\z) - \gamma  \geq {\tt Count}_{y}^{\tt O}(\z) + \gamma + \mathbb{I}(y<y_a).
\end{equation}
As a result, as long as $\gamma \leq \Gamma$ with $\Gamma = \frac{1}{2}({\tt Count}_{y_a}^{\tt O}(\z) - {\tt Count}_{y}^{\tt O}(\z) - \mathbb{I}(y<y_a))$, we have $F^t(\z) = F^{\tt O}(\z)$.
\end{proof}

Theorem \ref{thm.certified.inference.consistency} draws inspiration from the certified robustness defense against data poisoning attacks \cite{DBLP:conf/iclr/0001F21,wang2022improved,jia2021intrinsic}. However, straightforwardly applying the certified robustness from \cite{DBLP:conf/iclr/0001F21} results in a worse certification condition. In other words, \projectname ~would no longer certify on inference samples it should have.
In fact, Theorem \ref{thm.straitforward.certification} in Appendix \ref{appendix.straitforward} can be regarded as having the certification condition $2\gamma_1 + 2\gamma _2 + 2\gamma _3 \leq \Gamma$, which is apparently more strict than ours $2\gamma _1 + 0\gamma _2 + 1\gamma _3 \leq \Gamma$. As a result, our certification analysis is tailored to the machine unlearning under MLaaS and achieves a more fine-grained certification condition.

\section{Malicious Requests and Mitigation}
\label{sec.adversary.and.mitigation}
In addition to the toy attacks in Sec.\ref{subsec.toy.threats}, we present two new attacks that can be launched by queriers and data owners through malicious requests, aiming to consume more of the server's computational resources and slow down its response time. Experimental results are shown in Appendix \ref{appendix: malicious}.

\subsection{Hard-to-Classify Inference Requests}
\partitle{Attack}
In the first attack, we consider malicious inference requests from adversarial queriers, whose goal is to increase the overall inference latency by triggering the unlearning mechanism more frequently.
Adversarial queriers could submit numerous inference samples that are hard to classify, which are more likely to yield different prediction results across constituent models, increasing the likelihood of inconsistency, leading to additional unlearning execution. 
In particular, we consider malicious inference requests to contain random noise data.

%

\partitle{Mitigation}
There are two potential strategies for the server to mitigate the threat posed by hard-to-classify inference requests: detection before inference and discard after inference. For the former, the server screens all inference samples to detect whether they are malicious based on a predefined detection rule. For the latter, the server still runs the inference pipeline for the sample but discards it if the confidence score (e.g., softmax score) of the prediction result is exceptionally low. The discarded requests would be responded to with predefined unable-of-service message and won't trigger any unlearning execution.


\subsection{Scattered Unlearnings Requests}

\partitle{Attack}
In the second attack, we consider malicious unlearning requests from adversarial data owners, whose goal is to waste the server's computation on unlearning execution.
The server saves more computation on unlearning mechanism if more pending unlearning requests are processed all at once by a single unlearning execution. 
Adversarial data owners could collude and coordinate with each other to submit unlearning requests that are more scattered in timestamps and shards, which can hardly be executed together. 
This way, they increase the chance that a later unlearning request just arrived when the most recent unlearning execution finishes, which however could have been processed together with an earlier unlearning request. 
In particular, we consider malicious unlearning requests in a rolling manner: the adversarial data owners have to ensure that the $i$-th submitted unlearning request belongs to the $\{i\mod K\}$-th shard, where $K$ is the total number of shards. In other words, two unlearning request which belongs to the same shard that can be processed by a single update are placed as far as possible.

\partitle{Mitigation}
The server should prevent the attackers from establishing the correspondence between their training samples and the assigned shards
by shuffling and randomly dividing the training dataset into shards according to non-fixed hash functions. In this way, the adversary would have difficulty determining the location of their data.

\section{Additional Experiments} 
\label{appendix. experiments}

\begin{figure}[!t]
    \centering
    \includegraphics[width=0.45\textwidth]{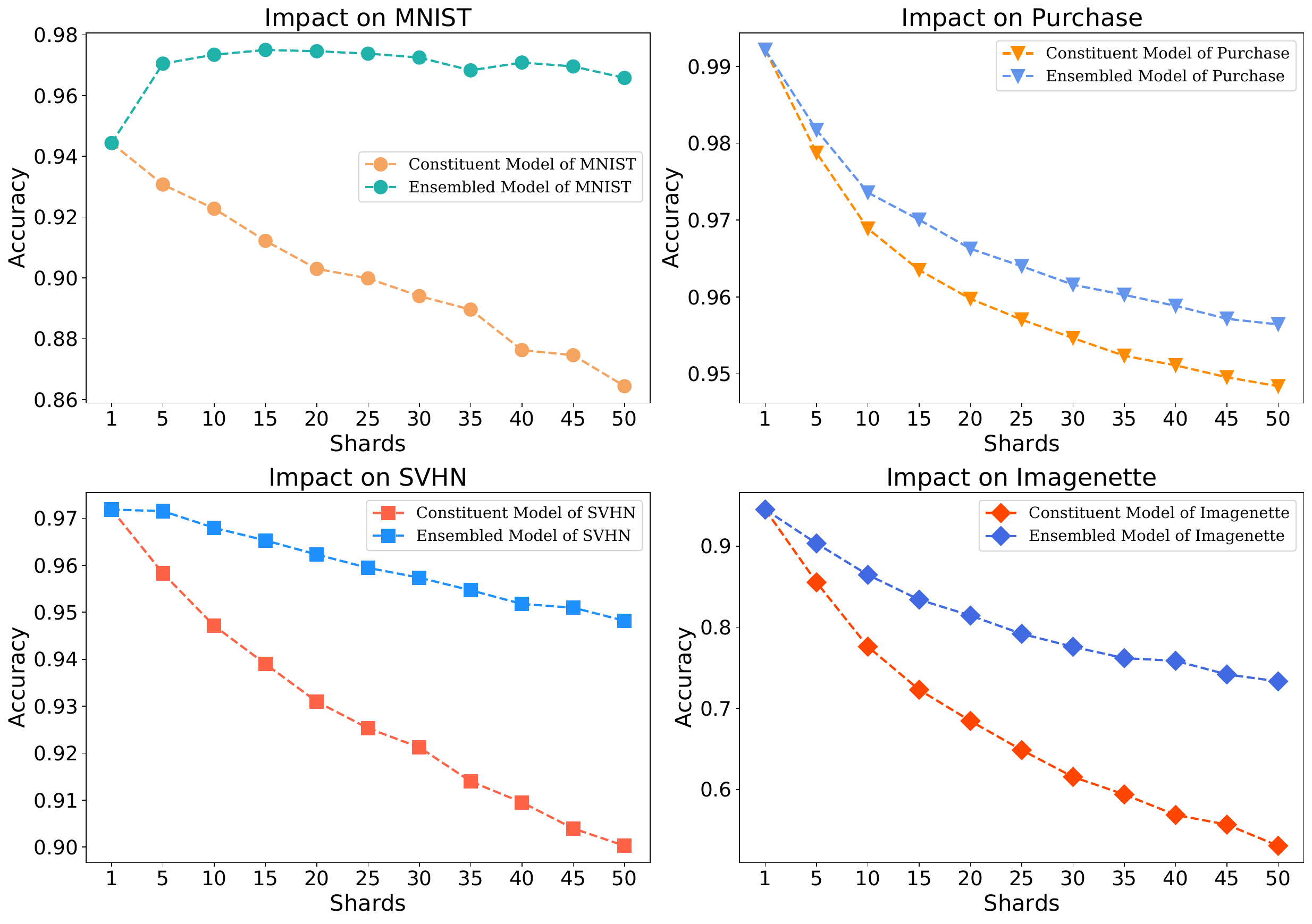}
    \caption{Evaluation on Model Accuracy with Different Number of Shards.}
    \label{fig:accuracy}
    \vspace{-1em}
\end{figure}

\subsection{Model Accuracy}
\label{appendix. accuracy}
We report the model accuracy with varying numbers of shards in Figure \ref{fig:accuracy}.
The accuracy of all models decreases as the number of shards increases.
Although each single constituent model exhibits a significant downward trend, the decline in the performance of the ensembled model is considerably more slight.
The performance of \projectname ~in Table \ref{tab: overview} varies across different datasets. Firstly, the larger the model, the longer the unlearning update takes, leading to longer AWT. 
Secondly, for easily classifiable datasets or models with higher accuracy, such as Purchase and Mnist, different constituent models are more likely to give consistent prediction results, resulting in a lower probability of uncertified prediction results. However, for Resnet-18 trained on ImageNette, which has lower accuracy, uncertified prediction results are more likely, triggering more frequent unlearning updates and longer waiting times.


\begin{figure}[!t]
    \centering
    \includegraphics[width=0.45\textwidth]{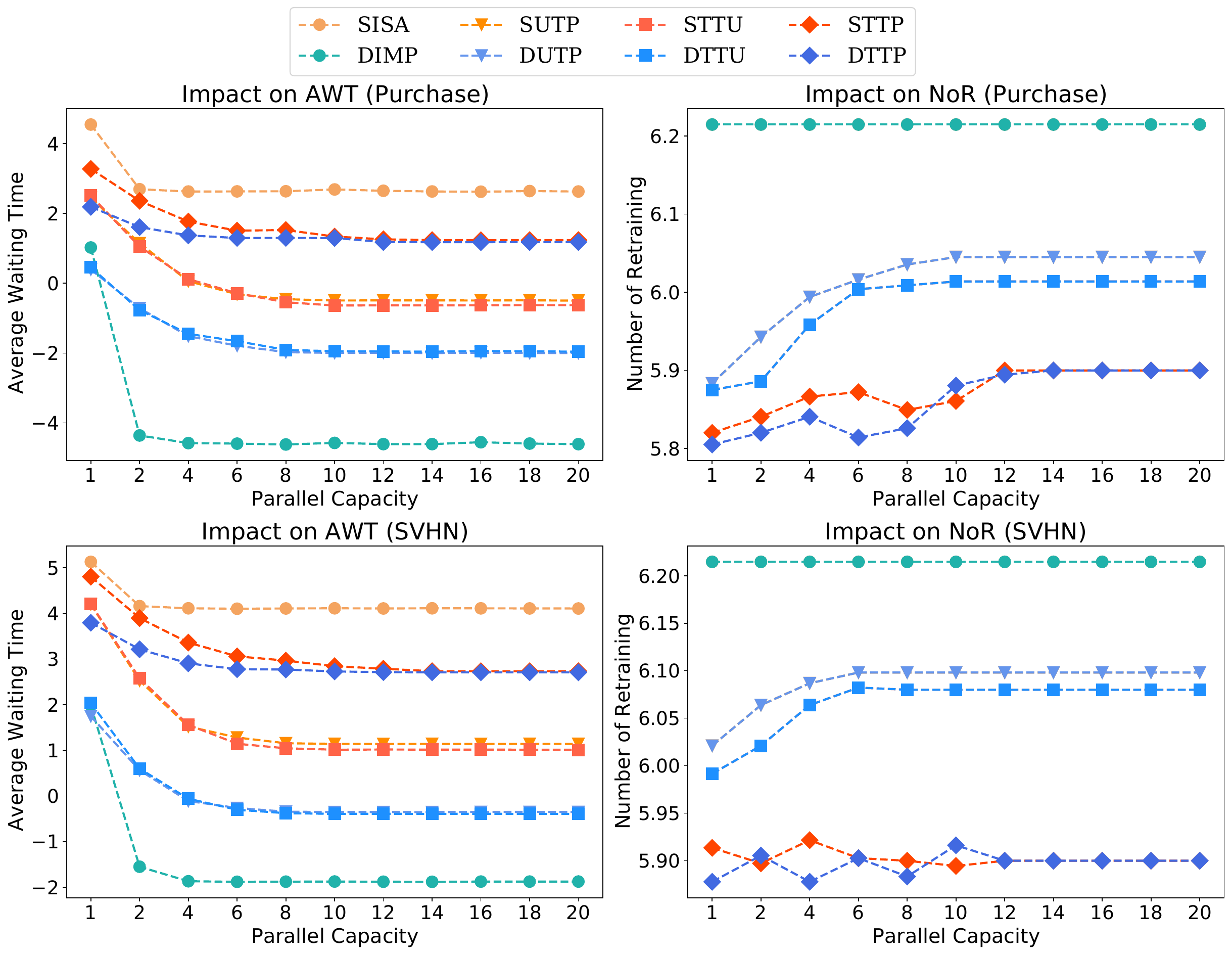}
    \caption{Evaluation on Parallel Capacity.}
    \label{fig:parameter_parallel}
    \vspace{-1em}
\end{figure}

\subsection{Evaluation of Parallel Capacity}
\label{appendix.parallel.capacity}
Considering the limited computing resources available for the server to retrain constituent models, we set a parallel capacity parameter representing the maximum number of constituent models that the server can retrain simultaneously.
As shown in Figure \ref{fig:parameter_parallel},
the inference latency for all methods is negatively correlated with parallel capacity.
However, once the parallel capacity exceeds half the number of constituent models, further increasing the parallel capacity has minimal effect.
This is because when the number of constituent models on the unlearning waiting list exceeds half of the total, all prediction results must be uncertified.
In this case, the unlearning update will start immediately, regardless of the method used. Therefore, it is rare for the server to have to retrain more than half of the constituent models simultaneously.
The NoR for most methods increases with the parallel capacity. 
This is because when the unlearning update is triggered, it is usually sufficient to retrain only part of the constituent models in the unlearning waiting list to make the prediction results certified. However, if the parallel capacity exceeds the number of models needing retraining, a few more models in the unlearning waiting list will be retrained in advance to fully utilize the server's computing resources. 

\begin{figure}[!t]
    \centering
    \includegraphics[width=0.45\textwidth]{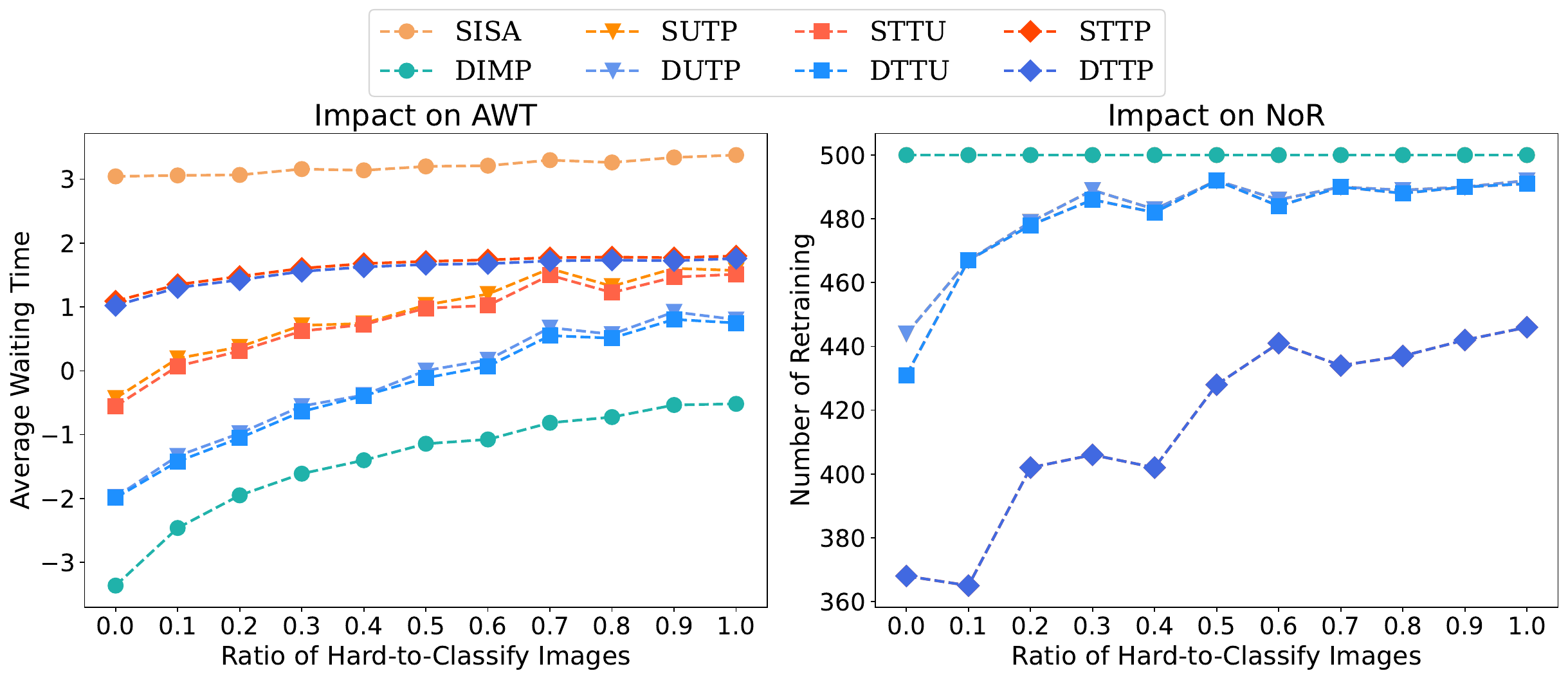}
    \caption{Evaluation on Hard-to-Classify Inference Requests.}
    \label{fig:counter_prediction}
    \vspace{-1em}
\end{figure}

\begin{figure}[!t]
    \centering
    \includegraphics[width=0.45\textwidth]{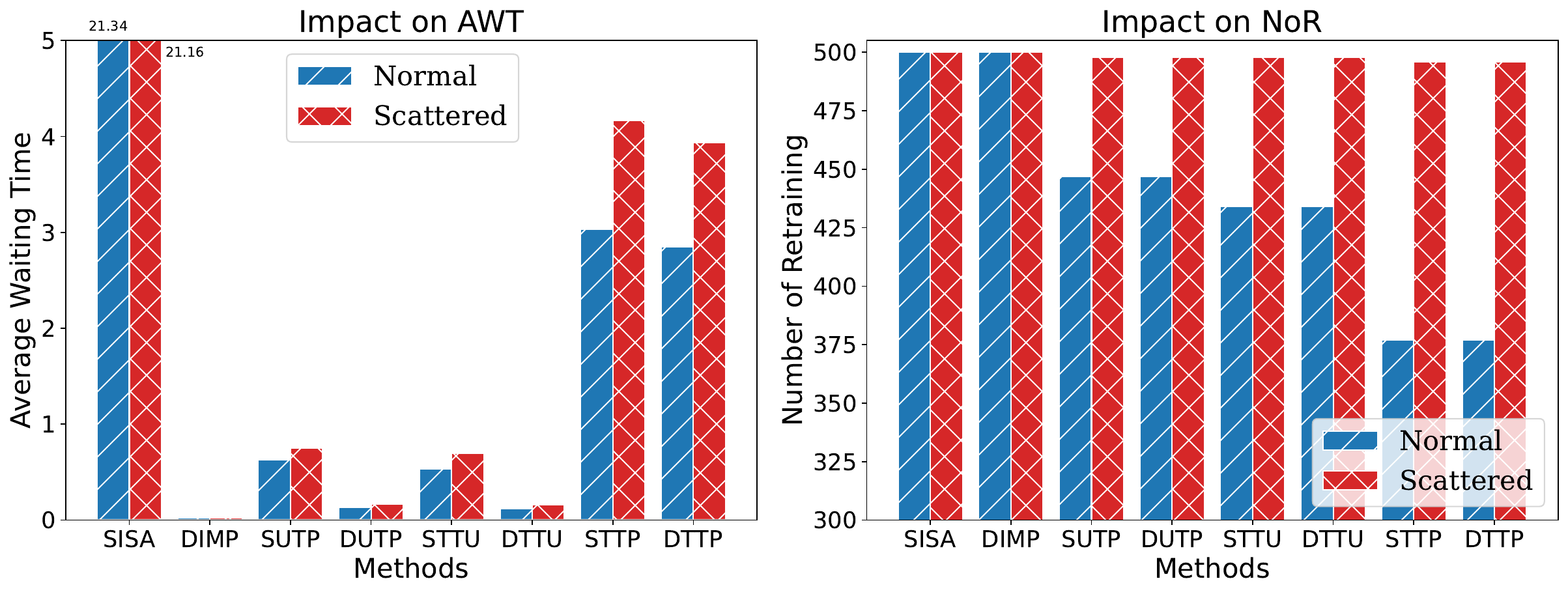}
    \caption{Evaluation of Scattered Unlearning Requests.}
    \label{fig:conuter_unlearn}
     \vspace{-1em}
\end{figure}

\subsection{Malicious Inference and Unlearning Requests}
\label{appendix: malicious}

\partitle{Hard-to-Classify Inference Requests}
We replace the test data of some inference requests with noisy images. 
The replacement ratio ranged from 10\% to 100\%. 
The experimental results in Figure \ref{fig:counter_prediction} demonstrate that submitting randomly generated inference requests is more likely to yield uncertified prediction results, which will trigger unlearning updates and increase inference latency. 
This attack can be mitigated by detecting and rejecting such requests.

\partitle{Scattered Unlearning Requests}
We compare the performance of \projectname ~with normal and scattered unlearning requests in Figure \ref{fig:conuter_unlearn}, which reveals that scattered unlearning requests can increase inference latency and cause computation overhead to approach the theoretical upper limit. 
Defending against such attacks can be achieved by shuffling the training dataset to prevent attackers from knowing the correspondence between data and shard indexes.
\balance
\end{document}